%% file: arxiv.tex
\documentclass[a4paper]{amsart}

\usepackage[left=2.5cm,right=2.5cm,bottom=2cm,top=2cm]{geometry}
\usepackage{amsmath,amsfonts,amsthm,amsaddr}
\usepackage{mathtools}
\usepackage{textcomp}
\usepackage{xcolor}
\usepackage{xspace}
\usepackage{stmaryrd}
\usepackage{tikz}
\usepackage{hyperref}

\usetikzlibrary{matrix,automata,shadows,shadings,shapes.symbols,arrows.meta}
\usetikzlibrary{decorations.pathmorphing}

\def\BibTeX{{\rm B\kern-.05em{\sc i\kern-.025em b}\kern-.08em
    T\kern-.1667em\lower.7ex\hbox{E}\kern-.125emX}}

		\AtBeginDocument{%
  \providecommand\BibTeX{{%
    Bib\TeX}}}

\input{./macros}

\allowdisplaybreaks

\begin{document}

\title[A Bisimulation-Inv.-Based Approach to the Separation of Poly.\ Compl.\ Classes]{A Bisimulation-Invariance-Based Approach to the Separation of Polynomial Complexity Classes
}


\author{Florian Bruse}
\address{Technical University of Munich, Germany}

\author{Martin Lange}
\address{University of Kassel, Germany}

\begin{abstract}
We investigate the possibility to separate the bisimulation-invariant fragment of P from that of NP, resp.\ PSPACE. We build on Otto's Theorem stating
that the bisimulation-invariant queries in P are exactly those that are definable in the polyadic $\mu$-calculus, and use a known construction from
model checking in order to reduce definability in the polyadic $\mu$-calculus to definability in the ordinary modal $\mu$-calculus within the
class of so-called power graphs, giving rise to a notion of relative regularity. We give examples of certain bisimulation-invariant queries in NP, 
resp.\ PSPACE, and characterise their membership in P in terms of relative non-regularity of particular families of tree languages. 
A proof of non-regularity for all members of one such family would separate the corresponding class from P, but the combinatorial complexity
involved in it is high. On the plus side, the step into the bisimulation-invariant world alleviates the order-problem that other approaches in descriptive
complexity suffer from when studying the relationship between P and classes above.  
\end{abstract}



	
\maketitle

\section{Introduction}
\input{./intro}

\section{Preliminaries}
\label{sec:prel}
\input{./prel}

\section{Relative Regularity}
\label{sec:relreg}
\input{./relreg}

\section{Power Graphs}
\label{sec:char}

\input{./char}

\section{Separability of Complexity Classes via Relative Non-Regularity of Tree Languages}
\label{sec:separate}
\input{./separate}

\section{Conclusion}
\label{sec:concl}
\input{./concl}

\bibliographystyle{plain}
\bibliography{./literature}

\end{document}

%% file: macros.tex
\renewcommand{\epsilon}{\varepsilon}
\renewcommand{\theta}{\vartheta}

\newcommand{\G}{\ensuremath{\mathcal{G}}}
\newcommand{\T}{\ensuremath{\mathcal{T}}}

\newcommand{\Cols}{\ensuremath{\mathsf{Col}}}

\def\newarrow#1{\mathop{{\hbox{\setbox0=\hbox{$\scriptstyle{#1\quad}$}{$%
\mathrel{\mathop{\setbox1=\hbox to
\wd0{\rightarrowfill}\ht1=3pt\dp1=-2pt\box1}\limits^{#1}}%
$}}}}}
\newcommand{\Transition}[3]{\ensuremath{#1 \newarrow{#2} #3}}
\newcommand{\TransitionD}[4]{\ensuremath{#1 \newarrow{#2}_{#4} #3}}

\newcommand{\Func}[1]{\ensuremath{\mathcal{F}_{#1}}}
\newcommand{\reset}[1]{\ensuremath{\mathit{rst}_{#1}}}

\newcommand{\kbisim}[1]{\ensuremath{\approx_{#1}}}

\newcommand{\Prop}{(\textsf{Prop})\xspace}
\newcommand{\Forth}{(\textsf{Forth})\xspace}
\newcommand{\Back}{(\textsf{Back})\xspace}

\newcommand{\mucalc}[1]{\ensuremath{\mathcal{L}_{\mu}^{#1}}}
\newcommand{\polymucalc}{\mucalc{\omega}}
\newcommand{\PHFL}[1]{\textsc{PHFL}\ensuremath{^{#1}}}
\newcommand{\trPHFL}[1]{\textsc{PHFL}\ensuremath{^{#1}_{\mathsf{tr}}}}

\newcommand{\Vars}{\ensuremath{\mathsf{Var}}}
\newcommand{\mudiam}[2]{\ensuremath{\langle #1 \rangle_{#2}}}
\newcommand{\mubox}[2]{\ensuremath{[ #1 ]_{#2}}}
\newcommand{\allbox}[1]{\ensuremath{[#1]^*}}
\newcommand{\mutrue}{\ensuremath{\mathtt{t\!t}}}

\newcommand{\repl}[2]{\ensuremath{\{#1{\shortleftarrow}#2\}}}

\newcommand{\sem}[3]{\ensuremath{[\![ #1 ]\!]^{#2}_{#3}}}

\newcommand{\mono}[1]{\ensuremath{\mathit{mono}(#1)}}
\newcommand{\poly}[1]{\ensuremath{\mathit{poly}(#1)}}

\newcommand{\PTIME}{\textsc{P}\xspace}
\newcommand{\NPTIME}{\textsc{NP}\xspace}
\newcommand{\PSPACE}{\textsc{PSpace}\xspace}
\newcommand{\EXPTIME}{\textsc{ExpTime}\xspace}
\newcommand{\EXPSPACE}{\textsc{ExpSpace}\xspace}
\newcommand{\NLOGSPACE}{\textsc{NLogSpace}\xspace}

\newcommand{\bisim}[1]{#1/\ensuremath{_{\sim}}}

\newcommand{\OneNonUnivNFA}[1]{\textsc{1NonUnivNFA}\ensuremath{^{#1}}\xspace}
\newcommand{\TwoNonUnivNFA}[1]{\textsc{2NonUnivNFA}\ensuremath{^{#1}}\xspace}

\newtheorem{theorem}{Theorem}
\newtheorem{proposition}[theorem]{Proposition}
\newtheorem{lemma}[theorem]{Lemma}
\newtheorem{corollary}[theorem]{Corollary}

\theoremstyle{plain}
\newtheorem{definition}[theorem]{Definition}
\newtheorem{example}[theorem]{Example}

\tikzstyle{every state}=[inner sep=0pt, minimum size=3mm]
\tikzset{initial text={}}
\tikzstyle{tnode}=[sharp corners,solid,thin,draw=black,matrix,matrix of nodes,inner sep=0pt,minimum size=2mm,
                    every node/.style={shape=rectangle,draw=none},column sep=0pt,ampersand replacement=\&]
\tikzstyle{snode}=[tnode,minimum size=1.5mm]
\tikzstyle{a0}=[red]
\tikzstyle{a1}=[blue]
\tikzstyle{rst0}=[red,dashed]
\tikzstyle{rst1}=[blue,dashed]
\tikzstyle{f0}=[fill=red]
\tikzstyle{f1}=[fill=blue]
\newcommand{\fzero}[2]{|[f0]|{\scriptsize\color{white} #1}\&|[fill=none]|{\scriptsize #2}\\}
\newcommand{\fone}[2]{|[fill=none]|{\scriptsize #1}\&|[f1]|{\scriptsize\color{white} #2}\\}
\newcommand{\fboth}[2]{|[f0]|{\scriptsize\color{white} #1}\&|[f1]|{\scriptsize\color{white} #2}\\}
\newcommand{\fnone}[2]{|[fill=none]|{\scriptsize #1}\&|[fill=none]|{\scriptsize #2}\\}

%% file: intro.tex

Without a doubt, the problem or determining whether \PTIME{=}\NPTIME is one of the, if not \emph{the} most famous and notorious open problem 
in complexity theory, theoretical computer science or even computer science as a whole. Despite ongoing effort for decades, 
there is still no proof either way, showing that any problem with polynomially verifiable witnesses (i.e.\ problem in \NPTIME)
also has polynomially computable solutions, or that some such problem cannot be solved by a polynomial-time algorithm.

The outstanding status of the \PTIME{=}\NPTIME question of course attracts a lot of attention, for instance in the form of proposed 
solutions that pop up regularly, typically claiming polynomial-time solvability of an NP-complete problem in rather short 
and sketchy papers. There is, however, also serious work on this question, for instance showing that certain proof methods
are unsuitable for settling the problem either way, cf.\ relativisation \cite{DBLP:journals/siamcomp/BakerGS75} or algebrisation 
\cite{DBLP:conf/stoc/AaronsonW08}.

This paper does not claim to solve the \PTIME{=}\NPTIME problem. We take some established results from the literature and work out
(some rather unwieldy) characterisations of languages of infinite trees such that a negative answer to the 
question of \PTIME{=}\NPTIME can be given by a proof of relative non-regularity for all members of a family of tree languages.
Relative non-regularity means that a proof via a pumping argument for instance would have to be confined to some other given class of tree-languages which makes
finding such a pumping argument even trickier. 

The possibility to settle the \PTIME{=}\NPTIME question by -- certainly complicated -- proofs of non-regularity of tree languages 
may seem surprising. 
However, 
%
the potential separation of complexity classes using arguments of formal expressive power is not new. It is at the heart 
of \emph{descriptive complexity theory}, the research field that aims at characterising 
computational complexity classes by means of the expressiveness of logical formalisms. It originates from Fagin's Theorem
\cite{Fagin74} equating \NPTIME with $\exists$SO, the existential fragment of Second-Order Logic. Characterisations of other
main time and space complexity classes have followed, for example \PSPACE{=}FO+PFP \cite{JACM::AbiteboulVV1997}, as well as
characterisations of $k$-\EXPTIME and $k$-\EXPSPACE for $k \ge 1$ through fragments of Higher-Order Logic with Least or Partial 
Fixpoints, cf.\ \cite{Imm89}. This opens up the principal possibility to separate complexity classes using tools of logical
inexpressivity but, to the best of our knowledge, the only known examples of successful proofs only re-establish results that 
are known otherwise already, for instance from the time or space hierarchy theorems \cite{HartmanisS65,FOCS::StearnsHL1965,Cook:STOC72}. 

Characterisations for well-known computational complexity classes by means of logics have successfully only been achieved for
classes containing \NPTIME. For classes below that, in particular \PTIME, one is faced with the \emph{order problem}: the input
to a computational device like a Turing Machine naturally comes with an order, as such inputs are written down on a Turing tape.
The ``input'' to a logical formula need not be ordered, though, i.e.\ a binary relation ordering the structure's elements either 
is or is not part of the structure itself. 

For characterisations of complexity classes from \NPTIME onwards, this makes no difference because the existence of a total order 
on the underlying structure can be formalised in $\exists$SO. Hence, a set of
ordered structures is definable in such logics iff their unordered versions are. 

For obtaining a logical characterisation of the complexity class \PTIME (and, with that, a possible approach to the settling of 
the \PTIME{=}\NPTIME question), the presence or absence of a total order does seem to make a difference, though. A logical 
characterisation of \PTIME (on general, \emph{unordered} structures) has not been found yet, 
despite active and ongoing research, cf.\ \cite{DBLP:conf/lics/Grohe08, DBLP:journals/siglog/DawarP24}.
The Immerman-Vardi Theorem \cite{Imm:relqcp,STOC::Vardi1982} provides a characterisation of those sets of structures
(otherwise also called \emph{queries}, \emph{languages}, \emph{problems}, etc.) that are \emph{order-invariant} and polynomial-time
computable in terms of FO+LFP -- First-Order Logic with Least Fixpoints.

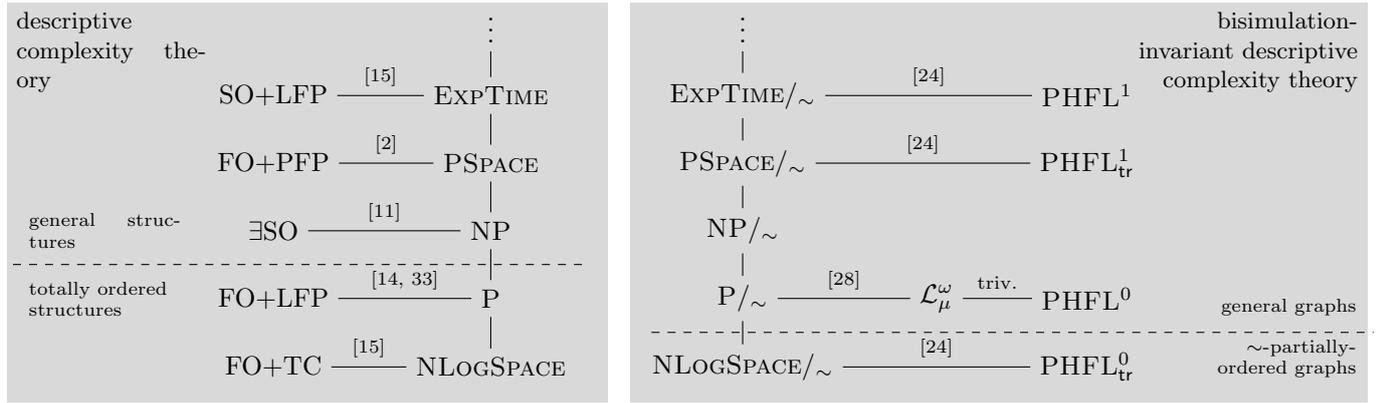
\begin{figure*}[t]
\begin{center}
  \begin{tikzpicture}

     \fill[gray!30] (-.6,-2.8) rectangle (9.1,2.5);
     \fill[gray!30] (-.9,-2.8) rectangle (-8.8,2.5);

     \node at (-8.8,2.5) [anchor=north west] {\parbox{2.5cm}{\small descriptive \\ complexity theory}};
     \node at (9.1,2.5) [anchor=north east] {\parbox{3cm}{\flushright \small bisimulation-invariant descriptive \\ complexity theory}};
     
     \matrix [row sep=3mm, column sep=9mm] (cc) {
         & \node (ccc) {$\vdots$}; & \node (bisccc) {$\vdots$}; & & \\  
         \node (solfp) {SO+LFP}; & \node (exptime) {\EXPTIME}; & 
           \node (bisexptime) {\bisim{\EXPTIME}}; & & \node (phfl1) {\PHFL{1}}; \\
         \node (fopfp) {FO+PFP}; & \node (pspace) {\PSPACE}; & 
           \node (bispspace) {\bisim{\PSPACE}}; & & \node (phfl1tr) {\trPHFL{1}}; \\
         \node (eso) {$\exists$SO}; & \node (nptime) {\NPTIME}; & 
           \node (bisnptime) {\bisim{\NPTIME}}; & & \\
         \node (folfp) {FO+LFP}; & \node (ptime) {\PTIME}; & 
           \node (bisptime) {\bisim{\PTIME}}; & \node (polymucalc) {\polymucalc}; & \node (phfl0) {\PHFL{0}}; \\
         \node (fotc) {FO+TC}; & \node (nlogspace) {\NLOGSPACE}; & \node (bisnlogspace) {\bisim{\NLOGSPACE}}; & & \node (phfl0tr) {\trPHFL{0}}; \\
           };
     
     \tikzset{every node/.style={font=\scriptsize}}
      
     \draw (solfp) edge node [above] {\cite{Imm:lanccc}} (exptime) 
             (bisexptime) edge node [above] {\cite{conf/ifipTCS/LangeL14}} (phfl1)
           (fopfp) edge node [above] {\cite{JACM::AbiteboulVV1997}} (pspace) 
             (bispspace) edge node [above] {\cite{conf/ifipTCS/LangeL14}} (phfl1tr)
           (eso) edge node [above] {\cite{Fagin74}} (nptime)
           (folfp) edge node [above] {\cite{Imm:relqcp,STOC::Vardi1982}} (ptime)
             (bisptime) edge node [above] {\cite{Otto/99b}} (polymucalc)
             (polymucalc) edge node [above] {triv.} (phfl0)
           (fotc) edge node [above] {\cite{Imm:lanccc}} (nlogspace) 
             (bisnlogspace) edge node [above] {\cite{conf/ifipTCS/LangeL14}} (phfl0tr);
           
     \draw (nlogspace) edge (ptime)
           (ptime) edge (nptime)
           (nptime) edge (pspace)
           (pspace) edge (exptime)
           (exptime) edge (ccc)
           (bisnlogspace) edge (bisptime)
           (bisptime) edge (bisnptime)
           (bisnptime) edge (bispspace)
           (bispspace) edge (bisexptime)
           (bisexptime) edge (bisccc);
     
     \path (bisptime) -- ++(-1.2,-.45) coordinate (bislinel);
     \draw[dashed,thin] (bislinel) -- ++(9.5,0) coordinate (bisliner);
      
     \path (ptime) -- ++(1.2,.45) coordinate (liner);
     \draw[dashed,thin] (liner) -- ++(-7.5,0) coordinate (linel);
     
     \path (linel) -- ++(.1,.1) node [anchor=south west] {\parbox{2cm}{general structures}};
     \path (linel) -- ++(.1,-.1) node [anchor=north west] {\parbox{2cm}{totally ordered \\ structures}};
     
     \path (bisliner) -- ++(-.1,.1) node [anchor=south east] {\parbox{2cm}{\flushright general graphs}};
     \path (bisliner) -- ++(-.1,.1) node [anchor=north east] {\parbox{2cm}{\flushright $\sim$-partially- \\ordered graphs}};

  \end{tikzpicture}
\end{center}
\caption{Some known logical characterisations of standard complexity classes and bisimulation-invariant counterparts.}
\label{fig:bisimcomplexity}
\end{figure*}

While finding a characterisation of \emph{unordered}, polynomial-time computable queries is certainly interesting not just for
the question of settling \PTIME{=}\NPTIME, it is worth noting that such a characterisation is not necessary, at least for showing
\PTIME$\ne$\NPTIME, since this ``only'' requires finding one query in \NPTIME that does not belong to \PTIME, and that query need
not be order-invariant. It is therefore equally interesting to provide logical characterisations of fragments of \PTIME, i.e.\ 
polynomial-time computable queries over classes of structures with stronger invariance properties. 

One such work in this direction should perhaps deserve better attention: Otto \cite{Otto/99b} has studied the class \PTIME/$_{\sim}$
of polynomial-time computable queries on labeled graphs that are \emph{bisimulation-invariant}. He showed that this is captured by 
the \emph{Polyadic $\mu$-Calculus} $\polymucalc$ \cite{AndersenPMC:1994}, an extension of the better known \emph{Modal $\mu$-Calculus}
\mucalc{} \cite{Kozen83}. The polyadic version was first studied by Andersen as a formal specification language in program verification 
\cite{AndersenPMC:1994};\footnote{This technical report seems to have disappeared from the internet.}
it was independently studied by Otto under the name ``higher-dimensional $\mu$-calculus'' \cite{Otto/99b}. 
Here we stick to the earlier and shorter name.

Otto's Theorem was also the starting point for the development of a descriptive \emph{bisimulation-invariant} complexity 
theory on graphs, now with known characterisations of the complexity classes \bisim{\PSPACE} \cite{conf/ifipTCS/LangeL14} and
\bisim{$k$-\EXPTIME} \cite{DBLP:journals/corr/abs-2209-10311} for $k \ge 1$ etc., in terms of fragments of Polyadic Higher-Order Fixpoint
Logic \PHFL{}, an extension of the modal $\mu$-calculus by polyadic predicates and by higher-order functions, cf.\ Fig.~\ref{fig:bisimcomplexity}. We refer to \cite{conf/ifipTCS/LangeL14} for a formal definition of \PHFL{} and the fragments 
mentioned in this hierarchy. Their concrete syntax and semantics will play no further role in the developments here.

Fig.~\ref{fig:bisimcomplexity} shows one of the true benefits of the study of bisimulation-invariant queries on graphs with 
respect to computational complexity: the order problem has shifted its position (and its nature). Otto's Theorem, i.e.\ 
\bisim{\PTIME}{=}\polymucalc, relies on the availability of an order on the bisimilarity-equivalence classes on a graph, i.e.\ 
a particular partial order on the graph's nodes. One of the key components in its proof is the observation that this order is
definable in \polymucalc. Hence, the identification of any bisimulation-invariant graph query from complexity class $\mathcal{C}$
that is not definable in \polymucalc{} would separate \bisim{\PTIME} from $\bisim{\mathcal{C}}$. Moreover, we have the following
connection regarding separations.

\begin{theorem} 
\label{thm:bisseparate}
\enspace
\begin{enumerate}
\item[a)] \PTIME{=}\NPTIME iff \bisim{\PTIME}{=}\bisim{\NPTIME},
\item[b)] \PTIME{=}\PSPACE iff \bisim{\PTIME}{=}\bisim{\PSPACE}.
\end{enumerate}
\end{theorem}

\begin{proof}
The ``only if''-direction in each case is trivial. The ``if''-direction is given by the fact that both \NPTIME and \PSPACE 
contain hard problems that are bisimulation-invariant, for example the non-universality problem for NFA over a one-letter, resp.\ 
two-letter alphabet (\OneNonUnivNFA{} / \TwoNonUnivNFA{}). I.e.\ we have $\OneNonUnivNFA{} \in \bisim{\NPTIME}$. Suppose that 
\PTIME{$\ne$}\NPTIME. By \NPTIME-hardness of \OneNonUnivNFA{}, we then have $\OneNonUnivNFA{} \not\in \PTIME$. Since 
\bisim{\PTIME}{$\subseteq$}\PTIME, we also have $\OneNonUnivNFA{} \not\in \bisim{\PTIME}$ and therefore 
\bisim{\PTIME}{$\ne$}\bisim{\NPTIME}. The reasoning for (b) is the same using \TwoNonUnivNFA{} instead. 
\end{proof}

So, not only would a separation of \bisim{\PTIME} from \bisim{\NPTIME}, resp.\ \bisim{\PSPACE} not suffer from the order problem
and lift to a separation of the respective non-bisimulation-invariant classes, but there is also a third reason for studying
the relationships on the bisimulation-invariant side: bisimulation-invariance reduces the question of definability of a set of
graphs in a logic (or any other formalism) to that of definability of a \emph{tree language}. This is a simple consequence of
the well-known fact that graphs and their tree unfoldings are bisimilar. Hence, in order to show that the set of graphs with 
some particular property is \emph{not} definable in a logic that is bisimulation-invariant, it suffices to show this for the
set of (unranked) trees with that property. And this allows machinery from formal (tree) language theory to be employed for
such tasks.

The rest of the paper is organised as follows. Sect.~\ref{sec:prel} introduces the necessary technical preliminaries: labeled graphs
and trees, bisimulations, the polyadic $\mu$-calculus, etc. It also recalls Otto's Theorem and a product construction which allows
the question of definability in \polymucalc{} to be reduced to that of definability in \mucalc{} within the set of so-called 
power graphs, leading to the notion of relative regularity introduced in Sect.~\ref{sec:relreg}.
Sect.~\ref{sec:char} then studies the structure of such power graphs in terms of languages of unranked, infinite trees.
Sect.~\ref{sec:separate} picks up on the two problems \OneNonUnivNFA{} and \TwoNonUnivNFA{} mentioned in the proof of Thm.~\ref{thm:bisseparate}
and characterises tree unfoldings of the respective power graphs. Together with the characterisation of power graphs from 
Sect.~\ref{sec:char} we then obtain two families $(L_d)_{d \ge 1}$ of languages of unranked, infinite trees such that their
non-regularity, relative to the language of tree unfoldings of $d$-powers of graphs, for all $d \ge 1$ is a witness for 
$\bisim{\PTIME} \ne \bisim{\NPTIME}$, resp.\ $\bisim{\PTIME} \ne \bisim{\PSPACE}$. The final step to such a separation along
these lines -- a formal proof of relative non-regularity of these tree languages -- will involve some difficult combinatorial 
arguments, though, and is therefore und unfortunately beyond the scope of this paper.

%% file: prel.tex

We write $[n]$ for the set $\{0,\dotsc,n-1\}$.

\subsection{Graphs, Trees, Bisimulations}

Let $\Cols$ be a finite, non-empty set of \emph{colours}, and $\Sigma$ be a finite, non-empty \emph{alphabet}.\footnote{In fact, 
it would suffice to restrict ourselves to cases where $1 \le |\Sigma| \le 2$.}
A \emph{$(\Sigma,\Cols)$-coloured, rooted graph} (or just \emph{$(\Sigma,\Cols)$-graph} for short) is a 
$\G = (V,\Transition{}{}{},L,v_I)$ such that 
$V$ is a set of nodes with a designated \emph{root} node $v_I \in V$. The relation 
$\Transition{}{}{} \subseteq V \times \Sigma \times V$ contains 
directed edges that are labeled with a symbol from $\Sigma$. We write $\Transition{v}{a}{v'}$ instead of 
$(v,a,v') \in \Transition{}{}{}$. 
The function $L : V \to 2^{\Cols}$ assigns a set of colours to each 
node. Hence, each edge has a unique colour, but nodes can have an arbitrary number of colours. When the names of colours and alphabet
symbols are irrelevant we may also simply speak of $(k_1,k_2)$-graphs for $k_1 = |\Sigma|$ and $k_2 = |\Cols|$.


A graph with a root node $v_I$, in which every node $v$ is reachable from $v_I$ on a unique path, is a tree. 
A \emph{$(\Sigma,\Cols)$-tree} is a $(\Sigma,\Cols)$-graph that happens to be a tree.

A \emph{bisimulation} between $(\Sigma,\Cols)$-graphs $\G = (V,\Transition{}{}{},L, v_I)$ and $\G' = (V',\Transition{}{}{},L',v'_I)$ is a 
binary relation $R \subseteq V \times V'$ such that for all $(u,v) \in R$ we have
\begin{itemize}
\item \Prop $L(u) = L'(v)$,
\item \Forth for all $a \in \Sigma$ and all $u' \in V$ such that $\Transition{u}{a}{u'}$ there is $v' \in V'$ such that 
      $\Transition{v}{a}{v'}$ and $(u',v') \in R$,  
\item \Back for all $a \in \Sigma$ and all $v' \in V'$ such that $\Transition{v}{a}{v'}$ there is $u' \in V$ such that 
      $\Transition{u}{a}{u'}$ and $(u',v') \in R$.  
\end{itemize}
Nodes $u,v$ are said to be \emph{bisimilar}, written $u \sim v$, if there is a bisimulation $R$ such that $(u,v) \in R$. 
Two $(\Sigma,\Cols)$-graphs $\G,\G'$ with root nodes $v_I$ and $v'_I$ are bisimilar, written $\G \sim \G'$, if $v_I \sim v'_I$.
A set $\mathfrak{G}$ of $(\Sigma,\Cols)$-graphs is said to be \emph{bisimulation-invariant} if for all $\G \in \mathfrak{G}$ 
and all $\G'$ such that $\G \sim \G'$ we have $\G' \in \mathfrak{G}$.

The \emph{unfolding} of a $(\Sigma,\Cols)$-graph $\G = (V,\Transition{}{}{},L,v_I)$ is the $(\Sigma,\Cols)$-tree 
$\G^{\mathsf{unf}} = (V^+,\Transition{}{}{},L^+,(v_I))$ 
where $V^+$ is the least set of finite, non-empty sequences over $V$ that contains $(v_I)$ and is closed under $\Transition{}{}{}$ 
in the following sense. If $(v_0,\ldots,v_n) \in V^+$ then for all $a \in \Sigma$ and all $v' \in V$ such that 
$\Transition{v_n}{a}{v'}$ in $\G$ we have $\Transition{(v_0,\ldots,v_n)}{a}{(v_0,\ldots,v_n,v')}$ in $\G^{\mathsf{unf}}$. 
At last, $L^+(v_0,\ldots,v_n) := L(v_n)$.

It is well-known that bisimulations cannot distinguish graphs from their tree unfoldings.

\begin{proposition}
\label{prop:treeunfoldbisim}
Let $\Sigma,\Cols$ be given and $\G$ be a $(\Sigma,\Cols)$-graph. Then $\G \sim \G^{\mathsf{unf}}$. 
\end{proposition} 

For a set $\mathfrak{G}$ of $(\Sigma,\Cols)$-graphs we write 
$\mathfrak{G}^{\mathsf{unf}} := \{ \G^{\mathsf{unf}} \mid \G \in \mathfrak{G} \}$ for the set of $(\Sigma,\Cols)$-trees
that are unfoldings of graphs in $\mathfrak{G}$.

\subsection{The Polyadic $\mu$-Calculus}

The polyadic modal $\mu$-calculus is interpreted over $(\Sigma,\Cols)$-graphs. It can be seen as a second-order (modal) logic
whose formulas define properties of tuples of nodes of some fixed arity $d$. 
Formally, the syntax of
the \emph{$d$-ary} modal $\mu$-calculus $\mucalc{d}$ is the following for given $\Sigma$ and $\Cols$ as above. 
\begin{displaymath}
\varphi \enspace := \enspace c_i \mid X \mid \varphi \wedge \varphi \mid \neg\varphi \mid \mudiam{a}{i}\varphi \mid \mu X.\varphi \mid 
           \sigma\varphi
\end{displaymath}
where $c \in \Cols$, $i \in [d] := \{0,\ldots,d-1\}$, $a \in \Sigma$, $\sigma: [d] \to [d]$ and $X$ is taken from some countably infinite set 
of variables $\Vars$. The polyadic $\mu$-calculus is 
$\polymucalc := \bigcup_{d \ge 1} \mucalc{d}$. The intuition for these formulas is that they talk
about $d$-tuples. $c_i$ expresses that the $i$-th component of such a tuple satisfies $c$ and
$\mudiam{a}{i}$ is the standard modal diamond for the $i$-th component of such a tuple. 
The \emph{replacement operator} $\sigma \varphi$, with $\sigma \colon [d] \to [d]$ being
a not necessarily injective mapping, 
expresses that the tuple obtained by re-arranging the current tuple according to $\sigma$
satisfies $\varphi$.

As is routinely done with the ordinary modal $\mu$-calculus, we 
assume that no two distinct fixpoint subformulas $\mu X.\psi$ of some formula $\varphi$ both use the same fixpoint variable
$X$. We also need to demand that each fixpoint variable $X$ occurs under an even number of negations in its then uniquely
defining fixpoint subformula $\mu X.\psi$.

Further Boolean, modal and fixpoint operators are introduced in the standard way: 
$\nu X.\varphi := \neg \mu X.\neg\varphi[\neg X/X]$, $\mubox{a}{i} := \neg \mudiam{a}{i}\neg\varphi$, 
$\varphi \vee \psi := \neg(\neg\varphi \wedge \neg\psi)$, $\mutrue := c_0 \vee \neg c_0$ for some $c$, etc. 
We will also use intuitive notation for functions $\sigma$ in the replacement operator. For instance, the
replacement of $i$ by $j$, leaving all other indices unchanged, will simply be denoted as $\repl{i}{j}$.

Given a $(\Sigma,\Cols)$-graph $\G = (V,\Transition{}{}{},L,v_I)$ and an assignment $\theta: \Vars \to 2^{V^d}$ of (fixpoint) 
variables to sets of $d$-tuples of nodes, a formula $\varphi$ defines such a set $\sem{\varphi}{\G}{\vartheta}$ in the 
following way. We write $\bar{v}_i$ to denote the $i$-th component of the $d$-tuple $\bar{v}$, i.e.\ 
$\bar{v} = (\bar{v}_0,\ldots,\bar{v}_{d-1})$. We also write $\bar{v}[i \leftarrow u]$ for the tuple that results from
$\bar{v}$ by replacing its $i$-th component with the node $u$.
\begin{align*}
\sem{c_i}{\G}{\vartheta} &:= \{ \bar{v} \in V^d \mid c \in L(\bar{v}_i) \} \\ 
\sem{X}{\G}{\vartheta} &:= \vartheta(X) \\
\sem{\varphi \wedge \psi}{\G}{\vartheta} &:= \sem{\varphi}{\G}{\vartheta} \cap \sem{\psi}{\G}{\vartheta} \\ 
\sem{\neg\varphi}{\G}{\vartheta} &:= V^d \setminus \sem{\varphi}{\G}{\vartheta} \\ 
\sem{\mudiam{a}{i}\varphi}{\G}{\vartheta} &:= \{ \bar{v} \in V^d \mid \exists u \in V \text{ s.t. } \Transition{\bar{v}_i}{a}{u}
  \text{ and } \\ &\hspace*{2.5cm} \bar{v}[i \leftarrow u] \in \sem{\varphi}{\G}{\vartheta} \} \\
\sem{\mu X.\varphi}{\G}{\vartheta} &:= \bigcap \{ T \subseteq V^d \mid \sem{\varphi}{\G}{\vartheta[X \mapsto T]} \subseteq T \}  \\
\sem{\sigma\varphi}{\G}{\vartheta} &:= \{ \bar{v} \mid 
   (\bar{v}_{\sigma(0)},\ldots,\bar{v}_{\sigma(d-1)}) \in \sem{\varphi}{\G}{\vartheta} \}
\end{align*}
where $\vartheta[X \mapsto T]$ denotes the variable assignment that maps $X$ to $T$ and any other variable $Y$ to 
$\vartheta(Y)$. 

As usual, variable assignments only need to be given for the free variables in a formula, and for formulas without
free (second-order) variables $X$ we can just write $\sem{\varphi}{\G}{}$ instead of $\sem{\varphi}{\G}{\vartheta}$ for
an arbitrary $\vartheta$. A set $\mathfrak{G}$ of $(\Sigma,\Cols)$-graphs is then said to be \emph{$\mucalc{d}$-definable}, if 
there is a closed formula $\varphi \in \mucalc{d}$ such that 
$\mathfrak{G} = L(\varphi) := \{ \G = (V,\Transition{}{}{},L,v_I) \mid (v_I,\ldots,v_I) \in \sem{\varphi}{\G}{} \}$. 
Two formulas $\varphi,\psi$ are \emph{equivalent}, written $\varphi \equiv \psi$, if $L(\varphi) = L(\psi)$.

Note that $\mucalc{1}$ is just the ordinary modal $\mu$-calculus whose formulas define monadic predicates in underlying 
graphs. In particular, when $d=1$ then $\sigma\varphi \equiv \varphi$, as the only possible replacement $\sigma$ of 
type $[1] \to [1]$ is the identity function. We therefore also write $\mucalc{}$ instead of $\mucalc{1}$. 

As with the ordinary $\mu$-calculus, there is also a game-theoretic semantics for $\polymucalc$ that is equivalent 
to the denotational semantics via the Knaster-Tarski Theorem, cf.\ \cite{Kna28,Tars55}, given above. It is well-known
that the semantics for \mucalc{} can also be given via \emph{parity games}, cf.\ \cite{Stirling95}, and this is
the case for each $\mucalc{d}$ for $d \ge 1$ as well. In these games, the players push $d+1$ tokens around; one
token on the subformula graph of a given formula, and $d$ tokens on the underlying graph representing a $d$-tuple
of nodes. The moves are straight-forward, for instance when the formula token is on a subformula of the form 
$\mudiam{a}{i}\psi$, then the existential player moves the $i$-th node token along an $a$-edge to a successor, 
and the formula token gets moved to $\psi$. The winner of infinite plays is determined by the fixpoint type of the
outermost fixpoint variable occurring infinitely often. For details, we refer to \cite{Lange:FICS15}. Here, a vague
intuitive notion of this game-theoretic semantics is sufficient in order to aid the understanding of the meaning
of given $\polymucalc$-formulas.
 

We call a formula of $\mucalc{d}$ \emph{$i$-rooted} for $i \in [d]$, if for all of its subformulas of the form $c_j$ or 
$\mudiam{a}{j}\psi$ we have $j \ne i$, and for all of its subformulas $\sigma\psi$ we have $\sigma = \repl{j}{i}$ for
some $j \ne i$. Intuitively, in an evaluation of an $i$-rooted formula in a tuple $(v_0,\ldots,v_{d-1})$, the $i$-th
component is only ever used in order to reset another component to the root of the underlying graph.

We make use of two established theorems about the polyadic $\mu$-calculus. The first one -- Otto's Theorem -- relates the 
notion of definability in $\polymucalc$ with that of recognisability in polynomial time. As usual, a set $\mathfrak{G}$ of 
$(\Sigma,\Cols)$-graphs is said to be \emph{polynomial-time recognisable} if there is a polynomial-time algorithm that,
given an arbitrary $(\Sigma,\Cols)$-graph, correctly determines its membership in $\mathfrak{G}$. For the worst-case
running time estimations, we assume that graphs are represented finitely either via standard adjacency lists or matrices.
Infinite graphs are allowed for as long as they can be represented finitely via explicit back-edges. From now on, we 
restrict our attention to finitely representable graphs. Note that tree unfoldings of finite graphs are finitely 
representable, and so are the $d$-dimensional products of finitely representable graphs introduced formally below.
 
\begin{proposition}[Otto's Theorem \cite{Otto/99b}]
\label{prop:otto}
Let $\Sigma,\Cols$ be given and $\mathfrak{G}$ be a set of (finitely representable) $(\Sigma,\Cols)$-graphs. The 
following are equivalent.
\begin{enumerate}
\item[a)] $\mathfrak{G}$ is recognisable in polynomial time and bisimulation-invariant.
\item[b)] $\mathfrak{G}$ is definable by a $d$-rooted formula of $\mucalc{d+1}$ for some $d \ge 1$.  
\end{enumerate}
\end{proposition}

The restriction on $d$-rootedness is seen by a close inspection of the theorem's proof where the additional $(d+1)$-th index 
is only used to mark the root of an underlying graph. An immediate consequence of Otto's
Theorem is therefore that every $\polymucalc$-formula is equivalent to a $d$-rooted $\polymucalc$ formula (of possibly higher
dimensionality). We can therefore assume, henceforth, that all formulas of any fragment $\mucalc{d+1}$ of $\polymucalc$ are 
$d$-rooted.

The second result about $\polymucalc$ that we rely on, concerns its model checking problem. Lange and Lozes \cite{LL-FICS12} 
have shown, based on Andersen's initial considerations \cite{AndersenPMC:1994}, that the model checking problem for  
$\polymucalc$ can be reduced to the model checking problem for $\mucalc{}$ using a particular product construction.

\subsection{Power Graphs}

For any $d \ge 1$ we consider the signature functor $\Func{d}$ that turns a pair $(\Sigma,\Cols)$ of sets of actions and
colours into the pair of sets 
\begin{displaymath}
\Func{d}(\Sigma,\Cols) := \big((\Sigma \times [d]) \cup \{ \reset{i} \mid 0 \le i < d \}, \Cols \times [d]\big)\ .
\end{displaymath}
We write an action $(a,i)$ simply as $a_i$ and, likewise, a colour $(c,i)$ simply as $c_i$.
  
\begin{definition}
Let $d \ge 1$. The \emph{$d$-(dimensional) product} of $d$ many $(\Sigma,\Cols)$-graphs 
$\G_i = (V_i,\TransitionD{}{}{}{i},L_i,v^i_I)$ with $i \in [d]$ is the 
$\Func{d}(\Sigma,\Cols)$-graph $\prod_{i=0}^{d-1} \G_i := (V_0 \times \dotsb \times V_{d-1},\Transition{}{}{},L,(v^0_I,\dotsc,v^{d-1}_I))$ where 
\begin{itemize}
\item $\Transition{(v_0,\ldots,v_{d-1})}{a_i}{(v'_0,\ldots,v'_{d-1})}$ if $\TransitionD{v_i}{a}{v'_i}{i}$ 
      and $v'_j = v_j$ for all $j \ne i$,
\item $\Transition{(v_0,\ldots,v_{d-1})}{\reset{i}}{(v'_0,\ldots,v'_{d-1})}$ if $v'_i = v^i_I$
      and $v'_j = v_j$ for all $j \ne i$,
\item $c_i \in L(v_0,\ldots,v_{d-1})$ if $c \in L^i(v_i)$
\end{itemize}
for $a \in \Sigma$, $c \in \Cols$ and $i \in [d]$.
\end{definition}

Note that $\prod_{i=0}^{d-1} \G_i$ can be seen as the asynchronous product of the $d$ graphs $\G_0, \dotsc, G_{d-1}$ where
action $a_i$ denotes an $a$-move in the $i$-th component with all other components remaining unchanged - hence the \emph{asynchronous}
product. Moreover, this product is enriched with additional transitions $\reset{i}$ that reflect the \emph{resetting} of 
the $i$-th component in a tuple to the root node of the underlying $i$-th graph. A tuple node in the product retains all
information about the colours of its components, encoded in the enlarged colour space. 

In the special case of $\G_0 = \dotsc = \G_{d-1} =: \G$ we call the $d$-product also the $d$-th \emph{power} of
$\G$ and write $\G^d$ instead of $\prod_{i=0}^{d-1} \G$. We write $\textsc{Power}_{d}$ for the language
of all $\Func{d}(\Sigma,\Cols)$-graphs that are bisimilar to a $d$-th power of some $(\Sigma, \Cols)$-graph
$\G$.

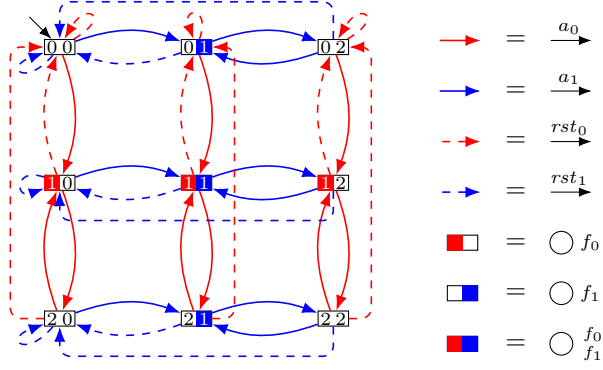
\begin{figure}[t]
\begin{center}
\begin{tikzpicture}[node distance=18mm, label distance=-1pt]

  \node[tnode] (00)               {\fnone{0}{0}}; 
  \node[tnode] (01) [right of=00] {\fone{0}{1}};  
  \node[tnode] (02) [right of=01] {\fnone{0}{2}}; 
  \node[tnode] (10) [below of=00] {\fzero{1}{0}}; 
  \node[tnode] (11) [right of=10] {\fboth{1}{1}}; 
  \node[tnode] (12) [right of=11] {\fzero{1}{2}}; 
  \node[tnode] (20) [below of=10] {\fnone{2}{0}}; 
  \node[tnode] (21) [right of=20] {\fone{2}{1}};  
  \node[tnode] (22) [right of=21] {\fnone{2}{2}}; 
  
  \foreach \y in {0,1,2} {
  \path[-Latex,semithick] (0\y) edge [a0, bend left=20] (1\y)
                          (1\y) edge [a0, bend left=20] (2\y)
                          (2\y) edge [a0, bend left=20] (1\y);
    }
  \foreach \x in {0,1,2} {
  \path[-Latex,semithick] (\x0) edge [a1, bend left=20] (\x1)
                          (\x1) edge [a1, bend left=20] (\x2)
                          (\x2) edge [a1, bend left=20] (\x1);
    }
  \path[semithick,rst0,-Latex] (10) edge [bend left=20] (00);
  \path[semithick,rst0,-Latex,rounded corners,draw] (20) -- ++(-.65,0) |- (00);
  \foreach \y in {1,2} {
    \path[semithick,rst0,-Latex] (1\y) edge [bend left=20] (0\y);
    \path[semithick,rst0,-Latex,rounded corners,draw] (2\y) -- ++(.5,0) |- (0\y);
  }
  \path[semithick,rst1,-Latex] (01) edge [bend left=20] (00);
  \path[semithick,rst1,-Latex,rounded corners,draw] (02) -- ++(0,.6) -| (00);
  \foreach \x in {1,2} {
    \path[semithick,rst1,-Latex] (\x1) edge [bend left=20] (\x0);
    \path[semithick,rst1,-Latex,rounded corners,draw] (\x2) -- ++(0,-.5) -| (\x0);
  }
  \path[semithick,rst0,-Latex] (00) edge [out=30,in=60,looseness=12]   (00)
                               (01) edge [out=70,in=110,looseness=15]  (01)
                               (02) edge [out=30,in=60,looseness=12]   (02);
  \path[semithick,rst1,-Latex] (00) edge [out=230,in=200,looseness=12] (00)
                               (10) edge [out=160,in=200,looseness=7]  (10)
                               (20) edge [out=230,in=200,looseness=12] (20);
  
  \path[Latex-] (00) edge ++(-.4,.4);

  \matrix[row sep=3mm,anchor=north] (legend) at (6,.4) {
    \node (l1) {$=$}; \\
    \node (l2) {$=$}; \\
    \node (l3) {$=$}; \\
    \node (l4) {$=$}; \\
    \node (l5) {$=$}; \\
    \node (l6) {$=$}; \\
    \node (l7) {$=$}; \\
  };
  
  \path[draw,semithick,a0,Latex-,shorten <=2mm]  (l1) -- ++(-1,0);
  \path[draw,-Latex,shorten <=2mm]  (l1) -- node [above=-1pt,pos=.6] {\scriptsize $a_0$} ++(1,0);

  \path[draw,semithick,a1,Latex-,shorten <=2mm] (l2) -- ++(-1,0);
  \path[draw,-Latex,shorten <=2mm]  (l2) -- node [above=-1pt,pos=.6] {\scriptsize $a_1$} ++(1,0);

  \path[draw,semithick,rst0,Latex-,shorten <=2mm]  (l3) -- ++(-1,0);
  \path[draw,-Latex,shorten <=2mm]  (l3) -- node [above=-1pt,pos=.6] {\scriptsize $\reset{0}$} ++(1,0);

  \path[draw,semithick,rst1,Latex-,shorten <=2mm] (l4) -- ++(-1,0);
  \path[draw,-Latex,shorten <=2mm]  (l4) -- node [above=-1pt,pos=.6] {\scriptsize $\reset{1}$} ++(1,0);

  \path (l5) -- +(-.7,0) node[tnode] {\fzero{}{}} -- +(.6,0) node[state,label=0:{\scriptsize $f_0$}] {};

  \path (l6) -- +(-.7,0) node[tnode] {\fone{}{}} -- +(.6,0) node[state,label=0:{\scriptsize $f_1$}] {};

  \path (l7) -- +(-.7,0) node[tnode] {\fboth{}{}} -- +(.6,0) node[state,label=0:{\small $f_0 \atop f_1$}] {};
  
\end{tikzpicture}
\end{center}
\caption{$\Func{2}(\{a\},\{f\})$-graph that is the $2$-power of the $(\{a\},\{f\})$-graph in Ex.~\ref{exm:power}.}
\label{fig:power}
\end{figure}

\vspace*{-3mm}
\begin{example}
\label{exm:power}
Consider the $(\{a\},\{f\})$-graph
\begin{tikzpicture}[every state/.style={font=\scriptsize,inner sep=2pt,minimum size=2mm}]
  \node[state]                            (0)              {$0$};
  \node[state,label=90:{\scriptsize $f$}] (1) [right of=0] {$1$};
  \node[state]                            (2) [right of=1] {$2$};
  
  \path[-Latex] (0) edge [bend left=20] (1)
                (1) edge [bend left=20] (2)
                (2) edge [bend left=20] (1);
                
  \draw[Latex-] (0) -- ++(-.5,0);
\end{tikzpicture}
where the edge labels are omitted for brevity since they necessarily all are `$a$'. It's 2nd power $\G^2$ is
depicted in Fig.~\ref{fig:power}.
\end{example}

This product construction facilitates one half of a reduction of the model checking problem for $\polymucalc$ to that
of $\mucalc{}$. The other half is given by a simple syntactic transformation on formulas. From a $d$-rooted 
$\mucalc{d+1}$-formula $\varphi$ over $(\Sigma,\Cols)$ we obtain its \emph{monofication} $\mono{\varphi}$ as the 
$\mucalc{}$-formula over $\Func{d}(\Sigma,\Cols)$ that is obtained by  
\begin{itemize}
\item regarding every atomic subformula $c_i$, i.e.\ colour $c \in \Cols$ indexed by some dimension $i \in [d]$, as the atomic 
      subformula $c_i$ from $\Cols \times [d]$, and 
\item successively replacing every operator $\mudiam{a}{i}$ by $\mudiam{a_i}{}$, and every operator $\repl{i}{d}$ by 
      $\mudiam{\reset{i}}{}$.
\end{itemize} 

\begin{proposition}[\cite{AndersenPMC:1994,LL-FICS12}]
\label{prop:langelozes}
Let $\Sigma,\Cols$ be given, $d \ge 1$, $\varphi \in \mucalc{d+1}$ be closed and $d$-rooted, and $\G$ be a 
$(\Sigma,\Cols)$-graph with root $v_I$. Then $(v_I,\ldots,v_I) \in \sem{\varphi}{\G}{}$ iff 
$(v_I,\ldots,v_I) \in \sem{\mono{\varphi}}{\G^d}{}$. 
\end{proposition}

\begin{example}
Take, for instance, $\varphi := \mudiam{a}{0}(f_0 \wedge \mudiam{a}{1}\mubox{a}{1}f_1) \in \mucalc{2}$. It is not 
hard to see that $\G \models \varphi$ for the graph $\G$ from Ex.~\ref{exm:power}. With tokens $0$ and $1$ placed
on the initial node, it is possible to push token $0$ to the middle node satisfying $f$ and then push token $1$
to a node, namely the right one, from which every further push of that token along one edge puts it on a node
satisfying $f$ as well.
 
Likewise, it is not hard to see that $\G^2 \models \mudiam{a_0}{}(f_0 \wedge \mudiam{a_1}{}\mubox{a_1}{}f_1) =: \mono{\varphi}$ 
when inspecting $\G^2$ as depicted in Fig.~\ref{fig:power}. From the initial node in the upper left corner there
is a red edge to node $v$ in the upper middle that satisfies $f_0$ such that from $v$ there is a blue edge to
node $v'$ in the lower middle from where all blue edges lead to nodes satisfying $f_1$, namely the node in the
middle only.  
\end{example}

We observe that monofication can be reversed as it is clearly an injective operation. Given a formula $\varphi \in \mucalc{}$
over $\Func{d}(\Sigma,\Cols)$, we obtain its \emph{polyfication} $\poly{\varphi}$ by reversing the replacements of
the modal operators in the monofication process, and regarding atomic subformulas $c_i$ now as $c$ indexed with $i$. 

It is tempting to assume that -- while monofication reduces model checking of polyadic formulas to the of ordinary
modal formulas -- one can simply reverse the process such that polyfication reduces model checking of ordinary
modal formulas to the of polyadic ones. This is not strictly the case, though. It only reduces model checking of
ordinary modal formulas \emph{over $d$-products of graphs} to model checking of polyadic formulas over the corresponding
underlying factor of this $d$-product. We state this here only for $d$-powers for the sake of simplicity.

\begin{lemma}
\label{lem:polyfication}
Let $\Sigma,\Cols$ be given, $d \ge 1$, $\varphi \in \mucalc{}$ over $\Func{d}(\Sigma,\Cols)$ be closed, and $\G_i$ 
for $i=0,\ldots,d-1$ be $(\Sigma,\Cols)$-graph, each with root $v^i_I$, that are all bisimilar to some some graph
$\G$ with root $v_I$. Then $(v^0_I,\ldots,v^{d-1}_I) \in \sem{\varphi}{\G^d}{}$ iff 
$(v_I,\ldots,v_I) \in \sem{\poly{\varphi}}{\G}{}$. 
\end{lemma}
\begin{proof}
From Prop.~\ref{prop:langelozes} and the observation that $\poly{\mono{\varphi}} = \varphi$ for 
$\varphi \in \polymucalc$ and $\mono{\poly{\varphi}}$ for $\varphi \in \mucalc{}$, we get that 
$(v^0_I,\ldots,$ $v^{d-1}_I) \in \sem{\varphi}{\G'}{}$ iff $(v_I,\ldots,v_I) \in \sem{\poly{\varphi}}{\G}{}$
where $\G' := \prod_{i=0}^{d-1} \G_i$. Then note that, if $\G_i \sim \G$ for all $i=0,\ldots,d_1$, then
$\G' \sim \G^d$ from which the lemma's claim follows by bisimulation-invariance of $\mucalc{}$. 
\end{proof}

%

%% file: relreg.tex

We recall a well-known result about the expressive power of the ordinary modal $\mu$-calculus, namely
that it coincides with that of finite automata interpreted over (unranked) trees.

\begin{proposition}[\cite{lncs164*313,focs91*368,JW:autmcr,Wilke:2001:BBMS}]
\label{prop:regularautomaton}
Let $\Sigma,\Cols$ be given and $\mathfrak{G}$ be a bisimulation-invariant set of $(\Sigma,\Cols)$-graphs. 
The following are equivalent.
\begin{enumerate}
\item[a)] $\mathfrak{G}$ is definable by a formula of $\mucalc{}$.  
\item[b)] $\mathfrak{G}^{\mathsf{unf}}$ is a regular tree language.
\end{enumerate}
\end{proposition}

Monofication may appear to reduce the problem of definability in $\polymucalc$ to that of definability in 
$\mucalc{}$, i.e.\ regularity. This is not the case, though. The combination of Prop.~\ref{prop:langelozes} and 
Lemma~\ref{lem:polyfication} facilitates a characterisation of definability in $\polymucalc$ by definability
of a corresponding set of $d$-power graphs in $\mucalc{}$. This leads us to the notion of relative regularity.

\begin{definition}
Let $\Sigma,\Cols$ be given, and let $R,L$ be sets of $(\Sigma,\Cols)$-trees. We say that $L$ is 
\emph{regular relative to $R$}, if there is a formula $\varphi_L \in \mucalc{}$ 
s.t.\ for all trees $\T \in R$ we have: $\T \in L$ iff $\T \models \varphi_L$.
\end{definition} 

Note that regularity of $L$ relative to $R$ is not the same as regularity of $R \cap L$ which would entail the
existence of an $\mucalc{}$-formula that
\begin{itemize}
\item for any tree $\T \in R$ correctly determines whether or not $\T \in L$, and
\item rejects any tree $\T \not\in R$.
\end{itemize}
Regularity of $L$ relative to $R$ is weaker in the sense that it allows the underlying formula to give
arbitrary answers regarding membership in $L$ on all trees not belonging to $R$. In
other words, it requires the existence of some regular $L' \supseteq L$ such that
$L' \cap R = L$, but $R$ need not be regular itself. 

\begin{example}
Let $\Sigma = \{a,b\}$, $\Cols=\{f\}$. Consider the set $R_{\mathsf{word}}$ of all $(\Sigma,\Cols)$-trees such that every
path from the root -- assuming that each node has at least one child -- is labeled with the same word
from $(2^\Cols \times \Sigma)^\omega$. It is not hard to see that $R_{\mathsf{word}}$ is non-regular.\footnote{Interestingly, 
$R_{\mathsf{word}}$ is definable in \trPHFL{1} \cite{LLVG:TCS:2014} and therefore bisimulation-invariant and 
checkable in \PSPACE, cf.\ Fig.~\ref{fig:bisimcomplexity}, which may make it an alternative to the language
considered in Sect.~\ref{sec:sepPSPACE}.}

Moreover, consider the set $L_{\mathsf{uni}}$ of all $(\Sigma,\Cols)$-trees for which there is an $n$ such that all 
nodes on level $n$ are labeled $f$. This is equally not hard to recognise as being non-regular.\footnote{A 
proof of non-regularity, resp.\ non-definability of $L_{\mathsf{uni}}$ has even been 
published a while ago, cf.\ \cite{Emerson:1987:UIT}. As a side note we remark that $L_{\mathsf{uni}}$ is also
definable in $\trPHFL{1}$, cf.\ \cite{conf/icla/Lange19}. It is not yet another alternative language for the 
construction in Sect.~\ref{sec:sepPSPACE} as $R_{\mathsf{word}}$ is; it is in fact the language that is used 
there.}

However, $L_{\mathsf{uni}}$ is in fact regular relative to $R_{\mathsf{word}}$ because on trees where every path
forms the same sequence of alternative state labels and actions, we have that \emph{all} nodes at some distance from
the root satisfy $f$ iff some \emph{node} does. Hence, within the language $R_{\mathsf{word}}$, the trees belonging
to $L_{\mathsf{uni}}$ are characterised by the $\mucalc{}$-formula $\mu X.f \vee \mudiam{a}{}X \vee \mudiam{b}{}X$.
\end{example}

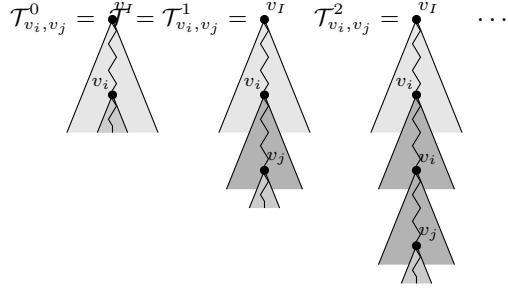
\begin{figure}
\begin{center}
\begin{tikzpicture}

  \tikzstyle{every node}=[draw,fill=black,shape=circle,inner sep=0pt,minimum size=3pt]
  \tikzstyle{zz}=[decorate,decoration={zigzag,amplitude=1.5pt}]
  
  \coordinate (1v0) at (0,0);
  \coordinate (1vi) at (0,-1);
  \coordinate (1vj) at (0,-2);
  
  \path[fill=gray!20] (1v0) -- ++(.6,-1.5) -- ++(-1.2,0) -- cycle;
  \draw (1v0) -- +(.6,-1.5) (1v0) -- +(-.6,-1.5);
  \path[fill=gray!60] (1vi) -- ++(.5,-1.25) -- ++(-1,0) -- cycle;
  \draw (1vi) -- +(.5,-1.25) (1vi) -- +(-.5,-1.25);
  \path[fill=gray!40] (1vj) -- ++(.2,-.5) -- ++(-.4,0) -- cycle;
  \draw (1vj) -- +(.2,-.5) (1vj) -- +(-.2,-.5);
     
  \node [label=45:{\scriptsize $v_I$}]  (1v0) at (1v0) {};
  \node [label=135:{\scriptsize $v_i$}] (1vi) at (1vi) {};
  \node [label=45:{\scriptsize $v_j$}]  (1vj) at (1vj) {};
  
  \draw[zz] (1v0) -- (1vi);
  \draw[zz] (1vi) -- (1vj);
  \draw[zz] (1vj) -- ++(0,-.5);

  \coordinate (0v0) at (-2,0);
  \coordinate (0vi) at (-2,-1);
  
  \path[fill=gray!20] (0v0) -- ++(.6,-1.5) -- ++(-1.2,0) -- cycle;
  \draw (0v0) -- +(.6,-1.5) (0v0) -- +(-.6,-1.5);
  \path[fill=gray!40] (0vi) -- ++(.2,-.5) -- ++(-.4,0) -- cycle;
  \draw (0vi) -- +(.2,-.5) (0vi) -- +(-.2,-.5);
     
  \node [label=45:{\scriptsize $v_I$}]  (0v0) at (0v0) {};
  \node [label=135:{\scriptsize $v_i$}] (0vi) at (0vi) {};
  
  \draw[zz] (0v0) -- (0vi);
  \draw[zz] (0vi) -- ++(0,-.5);

  \coordinate (2v0)  at (2,0);
  \coordinate (2vi)  at (2,-1);
  \coordinate (2vi2) at (2,-2);
  \coordinate (2vj)  at (2,-3);
  
  \path[fill=gray!20] (2v0) -- ++(.6,-1.5) -- ++(-1.2,0) -- cycle;
  \draw (2v0) -- +(.6,-1.5) (2v0) -- +(-.6,-1.5);
  \path[fill=gray!60] (2vi) -- ++(.5,-1.25) -- ++(-1,0) -- cycle;
  \draw (2vi) -- +(.5,-1.25) (2vi) -- +(-.5,-1.25);
  \path[fill=gray!60] (2vi2) -- ++(.5,-1.25) -- ++(-1,0) -- cycle;
  \draw (2vi2) -- +(.5,-1.25) (2vi2) -- +(-.5,-1.25);
  \path[fill=gray!40] (2vj) -- ++(.2,-.5) -- ++(-.4,0) -- cycle;
  \draw (2vj) -- +(.2,-.5) (2vj) -- +(-.2,-.5);
     
  \node [label=45:{\scriptsize $v_I$}]  (2v0)  at (2v0)  {};
  \node [label=135:{\scriptsize $v_i$}] (2vi)  at (2vi)  {};
  \node [label=45:{\scriptsize $v_i$}] (2vi2) at (2vi2) {};
  \node [label=45:{\scriptsize $v_j$}]  (2vj)  at (2vj)  {};
  
  \draw[zz] (2v0) -- (2vi);
  \draw[zz] (2vi) -- (2vi2);
  \draw[zz] (2vi2) -- (2vj);
  \draw[zz] (2vj) -- ++(0,-.5);

  \tikzset{every node/.style={fill=none,draw=none}}
  
  \node[anchor=east,font=\small] at (-.1,0)  {$\T = \T^1_{v_i,v_j} =$};
  \node[anchor=east,font=\small] at (-2.1,0) {$\T^0_{v_i,v_j} =$};
  \node[anchor=east,font=\small] at (1.9,0)  {$\T^2_{v_i,v_j} =$};
  
  \node at (3,0) {$\ldots$};
  
\end{tikzpicture}
\end{center}

\caption{The operation of pumping in trees as described in Def.~\ref{def:pump}.}
\label{fig:pumping}
\end{figure}

We now show that relative regularity enjoys a variant of the well-known
pumping lemma, formulated for tree languages. First
we define what it means to pump a tree. Note that if a graph
already is a tree, it is isomorphic to its own tree unfolding,
and we can refer to each node by the unique finite path from the
root that reaches it.

Let $\T$ be a tree with node set $V$ and initial node $v_I$. 
Let $\pi = v_I,v_1, v_2,\dotsc$ be a path in $\G$ and let $v_i, v_j$
with $i < j$ be two nodes on that path. Partition $V$ into three sets
as follows:
\begin{itemize}
\item $V_<$ is the set of $v_I,v'_1,\dotsc,v'_k$ that do not
contain the full sequence $v_I, v_1, \dotsc,v_i$ as a prefix,
\item $V_>$ is the set of $v_I,v'_1,\dotsc,v'_k$ that
do contain the full sequence $v_I,v_1,\dotsc,v_j$ as a prefix,
\item $V_p$ is the set of $v_I,v'_1,\dotsc,v'_{k}$ that contain
$v_I,v_1,\dotsc,v_i$ as a prefix, but not $v_I,v_1,\dotsc,v_j$.
\end{itemize}
Intuitively, $V_<$ is the set of nodes that either sit strictly above $v_i$
on $\pi$, or branch from $\pi$ strictly before $v_i$, the set $V_>$ is
the set of nodes that sit in $\pi$ below $v_j$ or branch from it after $v_j$,
and $V_p$ is the set of nodes that sit on $\pi$ between $v_i$ but strictly below $v_j$,
or branch from $\pi$ before $v_j$, but at or after $v_i$.

\begin{definition}
\label{def:pump}
Let $\T = (V, \Transition{}{}{}, L, v_I)$ be a $(\Sigma, \Cols)$-tree,
let $\pi = v_I,v_1,v_2,\dotsc$ be a path in $\T$, and let $v_i, v_j$ with
$i < j$ be two nodes on $\pi$.
The family $\T^k_{v_i,v_j}$ of $(\Sigma, \Cols)$- trees, one for each $k \geq 0$, 
i defined as 
$(V^k, \TransitionD{}{}{}{k}, L^k, v_I)$ where
\begin{itemize}
\item $V^k = (V \setminus V_p) \cup (V_p \times [k])$,
\item $L^k(v) = L(v)$ if $v \in V \setminus V_p$ and $L^k((v, k')) = L(v)$,
\item $\TransitionD{}{}{}{k}$ is defined via
\begin{itemize}
\item $\TransitionD{v}{a}{v'}{k}$ if $v, v' \in V \setminus V_p$ and $\Transition{v}{a}{v'}$,
\item $\TransitionD{(v,k')}{a}{(v', k')}{k}$ if $\Transition{v}{a}{v'}$ and $v, v' \in V_p$,
\item $\TransitionD{(v_{j-1}, k')}{a}{(v_i, k'+1)}{k}$ if $k' < k-1$ and $\TransitionD{v_{j-1}}{a}{v_j}{k}$,
\item if $k > 0$ then $\TransitionD{v_{i-1}}{a}{(v_i,0)}{k}$ if $\Transition{v_{i-1}}{a}{v_i}$ and further $\TransitionD{(v_{j-1}, k-1)}{a}{v_j}{k}$ if $\Transition{v_{j-1}}{a}{v_j}$,
\item if $k = 0$ then $\TransitionD{v_{i-1}}{a}{v_j}{0}$ if $\Transition{v_{i-1}}{a}{v_i}$.
\end{itemize} 
\end{itemize}
We say that $\T^k_{v_i,v_j}$ is obtained by pumping $k$ times between $v_i$ and $v_j$.
\end{definition}
An intuition for this operation can be found in Fig.~\ref{fig:pumping}: 
take from $\T$ the subgraph obtained by taking the 
path between $v_i$ inclusive and $v_j$ exclusive, as well as all its descendants that 
are not also descendants of $v_j$. Then either cut that subgraph (resulting in $\T^0_{v_i, v_j}$, left side),
or by repeating it $k$ times (resulting in e.g.\ $\T^2_{v_i,v_j}$, right side).

We now state the pumping lemma.

\begin{lemma}[Pumping Lemma for Relatively Regular Tree Languages]
\label{lem:pumping}
Let $\Sigma,\Cols$ be given and $L,R$ be sets of $(\Sigma,\Cols)$-trees. If $L$ is regular relative
to $R$ then there is an $n$ such that for every tree $\T \in L \cap R$ and every path $\pi = v_I,v_1,v_2,\ldots$
starting in the root node $v_I$, there are $i < j \leq n$ such for all $k\geq 0$, the tree 
$\T^k_{v_i,v_j}$ belongs to $L$ if it belongs to $R$.
\end{lemma}
Note that this differs from typical pumping lemmas in the sense that it is permitted
for trees that result from pumping to not be members of $L$, as long as they also are not members of $R$.
In other words, any proof of relative non-regularity requires the pumping to happen within the context of $R$. This makes successful 
applications of this form of Pumping Lemma for showing non-regularity relative to some language $R$ more difficult 
since a candidate tree $\T$ needs to be chosen such that the tree $\T'$, obtained by repeating or removing parts 
in $\T$, does not only fall outside of $L$ but remains inside $R$. 

Before we show Lem.~\ref{lem:pumping}, we introduce the notion of \emph{tree automata}, following
the presentation in \cite{Wilke:2001:BBMS} with small adjustments, mainly to adjust
to the multi-colour setting. 

\begin{definition}
\label{def:sta} Fix an alphabet $\Sigma$ and a set of colours $\Cols$. An alternating (symmetric) parity tree automaton (APT) is an $\mathcal{A} = (Q, q_I, \delta, \Omega)$
where $Q$ is a finite set of \emph{states} with $q_I \in Q$ the \emph{initial state}, $\Omega \colon Q \to \mathbb{N}$ the \emph{priority labeling} and $\delta\colon Q \to \mathcal{F}$ the \emph{transition function} 
where $\mathcal{F}$ is the set of formulas derived from the following grammar:
\[
\mathcal{F} \Coloneqq c \mid \neg c \mid \mudiam{a}{} q \mid \mubox{a}{} q \mid q \vee q \mid q \wedge  q
\]
where $c \in \Cols, a \in \Sigma, q \in Q$.
\end{definition}
Let $\G = (V, \Transition{}{}{}, L, v_I)$ be a $(\Sigma, \Cols)$-graph. Acceptance
of $\mathcal{A}$ is defined as a two-person game between $\exists$ and $\forall$,
where $\exists$ tries to show that $\mathcal{A}$ accepts $\G$ and $\forall$
tries to show the opposite. The positions of this game are of the form $(V \times Q)$,
with $(v_I, q_I)$ the initial position. Plays start in the initial position
and are extended by one of the players as follows:
\begin{itemize}
\item Plays in positions $(v, c)$ and $(v, \neg c)$ are over. Player
$\exists$ wins the first kind if $c \in L(v)$, and the second kind if $c \notin L(v)$,
otherwise $\forall$ wins.
\item In positions of the form $(v, \mudiam{a}{} q)$, player $\exists$
picks $v'$ with $\Transition{v}{a}{v'}$ and the play continues in $(v', q)$.
\item In positions of the form $(v, \mubox{a}{} q)$, player $\forall$
picks $v'$ with $\Transition{v}{a}{v'}$ and the play continues in $(v', q)$.
\item In positions of the form $(v, q_1 \vee q_2)$, player $\exists$ picks 
$i \in \{1,2\}$ and the play continues in $(v, q_i)$.
\item In positions of the form $(v, q_1 \wedge q_2)$, player $\forall$ picks 
$i \in \{1,2\}$ and the play continues in $(v, q_i)$.
\end{itemize}
A player that is stuck for lack of successors loses the game immediately.
An infinite play $(v_I, q_I), (v_1, q_1), (v_2, q_2), \dotsc$ induces a sequence
$\Omega(v_I), \Omega(v_1), \Omega(v_2), \dotsc$. Player $\exists$ wins the play
if the highest number that occurs infinitely often in this sequence is even,
otherwise $\forall$ wins. It is a standard result that in each position,
exactly one of the players has a positional winning strategy,
i.e.\ one where the choice in a position $(v, q)$ only depends on $v$ and $q$ \cite{focs91*368}.
We define the language of graphs accepted by $\mathcal{A}$, written $\mathcal{L}(\mathcal{A})$,
to be the set of all $(\Sigma, \Cols)$-graphs on which $\exists$ has a positional winning 
strategy from $(v_I, q_I)$. 

\begin{proof}[Proof of Lemma~\ref{lem:pumping}]
Suppose that $L$ is regular relative to $R$. Then there is $\varphi_L$ and, by 
Prop.~\ref{prop:regularautomaton} there is $\mathcal{A}_L = (Q, q_i, \delta, \Omega)$, 
such that for all $\T \in R$ we have that $\T \in L$ iff $\T \in \mathcal{L}(\mathcal{A}_L)$.

Let $n = 2^{|Q|} + 1$.
Let $\T = (V, \Transition{}{}{}, L,v_i) \in R \cap L$ be a tree, let $\pi = v_I,v_1,\dotsc$
be a path in $\T$. Since $\T \in L(\mathcal{A}_L)$, player
$\exists$ has a positional winning strategy from $(v_I, q_I)$ in the acceptance game. Now consider
the sequence of sets $S_1, S_2, \dotsc$ where $S_i = \{q \mid \exists \text{ has pos.\ win.\ str.\
from } (v_i, q)\}$. By the pigeonhole principle, there are $i, j \leq n$
with $i < j$ and $S_i = S_j$.\footnote{In fact, $S_i \subseteq S_j$ suffices,
but this requires the same upper bound on $n$.} 

We now claim that $\exists$ has a positional winning strategy 
for all $\T^k_{v_i,v_j}$. 
Independently of $k$, any position
of the form $(v, q)$ with $v \in V_>$ can only be reached via visiting
a position of the form $(v_j, q)$ for some $q$, whence it suffices to show 
that $\exists$ can enforce that this only happens for $q \in S_j$; after
that she can play out her strategy from $\T$, as the trees below $v_j$
in $\T$ and $\T^k_{v_i,v_j}$ are isomorphic and that strategy is positional.

Also, independently of $k$, we let $\exists$ play out her strategy from $\T$
on the part of $\T^k_{v_i,v_j}$ that looks like $\T_<$. 
This means that she wins all plays that stay in $\T_<$ forever.

Let $k = 0$.
Any play leaving $\T_<$ in $\T$ would do so via a move from $(v_{i-1}, q)$
to $(v_j, q')$ for some $q, q'$ with $q' \in S_i$, since by assumption 
on $\exists$'s strategy in $\T_<$, we have $q \in S_{i-1}$. Note that
this move is not necessarily triggered by $\exists$. Hence, any
move out of $T_<$ in $T^0_{v_i,v_j}$ would lead to $v_j$, and since 
$S_i = S_j$, the play continues in $(v_j, q')$ from which $\exists$
has a positional winning strategy. Hence $\T^0_{v_i,v_j} \in L(\mathcal{A}_L)$.

Now let $k > 1$. By the same argument, any move from $(v_{i-1}, q)$
towards some $(v', q') \notin \T_<$ would necessarily
lead to $((v_i,0), q')$ with $q' \in S_i$. We now let $\exists$ play the following
strategy: In positions of the form $((v, k'), q)$, she
plays the same way as in $(v, q)$, with moves to $(v_j, q')$ being
replaced by  moves to $((v_i, k'+1), q')$ if $k' < k-1)$.
Note that this means the following:
If the play reaches a position $((v_l, k'), q)$ with $i \leq l < j$,
then $q \in S_l$. By the same argument as above, this also
means that any play that reaches $(v_j, q)$ must be such that $q \in S_j$,
whence it is won by $\exists$.

We argue that this strategy wins all play for $\exists$. Since
all plays that stay in $\T_<$ forever, or reach $T_>$ eventually
are won by $\exists$ (see above), it suffices to argue about plays
that leave $\T_<$ eventually, but never reach $\T_>$. Hence,
such a play will eventually stay within one of the components
of $T_p \times [k-1]$, i.e.\ it happens in a copy of $\T_p$. Since
$\exists$ plays her winning strategy from the original copy, 
she wins any such play.

The above establishes that $\T^k_{v_i,v_j} \in \mathcal{L}(\mathcal{A}_L)$. 
By the definition of relative regularity, all such trees that are
also in $R$ hence must be members of $L$.
\end{proof}

Readers familiar with pumping lemmas for regular languages would probably have
expected that $n$ is linear in the size of $Q$, not exponential. The reason
for this discrepancy is the automaton model: since we work with unranked, unordered
trees, symmetric automata are the model of choice. However, these require an explicit
conjunction in the transition function (as opposed to the implicit conjunction present
in e.g., nondeterministic automata for ranked trees via transitions
sending copies into several directions at once). This also means that any notion
of a run may potentially contain several copies of the automaton for the same node
of the tree, but in different states. Hence, a pumping argument requires a powerset 
argument here.

%% file: char.tex
In Sect.~\ref{sec:separate} we will make use of the results presented in Sect.~\ref{sec:prel} in order to obtain a
characterisation of the separation of \PTIME from \NPTIME, resp.\ \PSPACE in terms of relative non-regularity of certain
tree languages. The languages providing that relative context are going to be the sets of trees that are bisimilar to
a power graph. The purpose of this section is therefore to provide a characterisation of such power graphs which could
be useful in potential applications of Prop.~\ref{lem:pumping} in order to carry out proofs of relative non-regularity,
i.e.\ do pumping within the class of trees representing power graphs.  

\subsection{A Characterisation via $d$-Bisimulations}

\begin{definition}
Let $\Sigma,\Cols$ and $d \ge 1$ be given and $\G = (V,\Transition{}{}{},L,v_I)$ be an $\Func{d}(\Sigma,\Cols)$-graph. 
We call a family $(\kbisim{ij})_{0 \leq i, j \leq d-1}$ of $d^2$ many relations a 
\emph{$d$-dimensional asynchronous bisimulation} or just \emph{$d$-bisimulation} for short, if it satisfies the 
following for all $i,j \in [d]$.
\begin{enumerate}
\item \Prop For all $v, v' \in V$ with $v \kbisim{ij} v'$, all $c \in \Cols$, we have $c_i \in L(v)$ iff 
      $c_j \in L(v')$.
\item \Forth For all $v_1, v_2 \in V$ with $v_1 \kbisim{ij} v_2$, all $a \in \Sigma$: 
      if there is $v'_1$ with $\Transition{v_1}{a_i}{v'_1}$ then there is $v'_2$ with 
		  $\Transition{v_2}{a_j}{v'_2}$ and $v'_1 \kbisim{ij} v'_2$. 
\item \Back For all $v_1, v_2 \in V$ with $v_1 \kbisim{ij} v_2$, all $a \in \Sigma$: 
      if there is $v'_2$ with $\Transition{v_2}{a_j}{v'_2}$ then there is $v'_1$ with 
			$\Transition{v_1}{a_i}{v'_1}$ and $v'_1 \kbisim{ij} v'_2$. 
\end{enumerate}
\end{definition}

A $d$-bisimulation on an $\Func{d}(\Sigma,\Cols)$-graph $\G$ can be seen as the extension
of the concept of ordinary bisimulation to a $d$-dimensional setting. When treating
nodes in $\G$ as if they were actually $d$-tuples, the $d$-bisimulation identifies
which components of these tuples behave similarly.  Intuitively, $v \kbisim{ij} v'$ if $v_i \sim v'_j$, i.e.\ the $i$-th (candidate) component of $v$ (as a $d$-tuple) is bisimilar to the $j$-th (candidate) component of $v'$. 

It is not hard to see that $d$-bisimulations are closed under arbitrary (pointwise) unions, and
that they always exist (by setting $v \kbisim{ii} v$ for all $v$ and $v \not\kbisim{ij} v'$ for $i \not = j$ and
$v \ne v'$). Hence, on
any given $\Func{d}(\Sigma,\Cols)$-graph there is a unique largest such family of relations. From now on, 
$(\kbisim{ij})_{0 \leq i, j \leq d-1}$ refers to this largest $d$-bisimulation. The following is shown by routine
inspection of the definition of $d$-bisimulations.

\begin{lemma}
\label{lem:kbisim-eq}
Let $\Sigma,\Cols,d$ be given and $\G$ be a $\Func{d}(\Sigma,\Cols)$-graph with associated $d$-bisimulation 
$(\kbisim{ij})_{0 \leq i, j \leq d-1}$. The following hold for all $i,j,h \in [d]$.
\begin{enumerate}
\item (Pseudo-Reflexivity) For all $v \in V$ we have $v \kbisim{ii} v$.
\item (Pseudo-Symmetry) For all $v, v' \in V$ we have $v \kbisim{ij} v'$ iff $v' \kbisim{ji} v$.
\item (Pseudo-Transitivity) For all $v, v', v'' \in V$ with $v \kbisim{ij} v'$ and $v' \kbisim{jh} v''$ we have 
       $v \kbisim{ih} v''$.
\end{enumerate}
\end{lemma} 

Remember that we are interested in characterising $d$-powers, i.e.\ $d$-dimensional products of the same underlying
graph. Such $d$-bisimulations, however, only witness the fact that a graph is bisimilar to something that could be
construed to be the Cartesian product of $d$ possibly different graphs, and this product also need not respect the
asynchronicity properties of $d$-products.

\begin{definition}
\label{def:persistence-roots}
Let $\Sigma,\Cols,d$ be given, $\G$ be an $\Func{d}(\Sigma,\Cols)$-graph and $(\kbisim{ij})_{0 \leq i, j \leq d-1}$ 
be the $d$-bisimulation associated with $\G$. 

We say that $(\kbisim{ij})_{0 \leq i, j \leq d-1}$ is \emph{persistent} if for all $v,v' \in V$ and $a \in \Sigma$: 
if $j \ne i$ and $\Transition{v}{a_i}{v'}$ or $\Transition{v}{\reset{i}}{v'}$ then $v \kbisim{jj} v'$.
 
Moreover, $(\kbisim{ij})_{0 \leq i, j \leq d-1}$ has the \emph{reset property} if for all $v,v' \in V$: 
if $\Transition{v}{\reset{i}}{v'}$ then $v' \kbisim{ii} v_I$.

$\G$ is a \emph{power} if $v_I \kbisim{ij} v_I$ for all $0 \leq i, j \leq d-1$.
\end{definition}
Persistence is a sanity criterion that holds on real products of $(\Sigma, \Cols)$-graphs:
advancing one component of a tuple will leave the other components unchanged. 
The reset property means that resets in the same component
of the candidate product always lead to nodes that behave like the root. Clearly,
if $\G = \prod_{i=0}^{d-1} \G_i$ for $(\Sigma,\Cols)$-graphs $\G_i$, then $(\kbisim{ij})_{0 \leq i,j <d}$ is automatically persistent and has the reset property. 

Finally, $\G$ is a power if the initial nodes of
all the candidate product members are bisimilar, whence $\G$ can be considered to be the $d$-power
of the same graph. Clearly, this holds on $\G$ if $\G \in \textsc{Power}_{d}$.

Now let $\Sigma,\Cols,d$ be given and $\G = (V,\Transition{}{}{},L,v_I)$ be an $\Func{d}(\Sigma,\Cols)$-graph 
with associated $d$-bisimulation $(\kbisim{ij})_{0 \leq i, j \leq d-1}$ that is persistent and has the reset property.
We define $d$ different $(\Sigma, \Cols)$-graphs $\G^i/d := (V/d, \TransitionD{}{}{}{i}, L^i/d, v^i_I)$ for $i \in [d]$ as follows.
\begin{itemize}
\item $V/d := V/_{\kbisim{ii}}$,
\item $\TransitionD{[v]_{\kbisim{ii}}}{a}{[v']_{\kbisim{ii}}}{i}$ iff there is $v'' \in [v']_{\kbisim{ii}}$ s.t.\ 
      $\Transition{v}{a_i}{v''}$,
\item $c \in (L^i/d)([v]_{\kbisim{ii}})$ iff $c_j \in L(v)$, 
\item $v^i_I := [v_I]_{\kbisim{ii}}$.
\end{itemize}

Each such $\G^i/d$ is indeed well-defined. 
Property \Prop gives us well-definedness of the labelling function $L^i/d$. 
The transition relation $\TransitionD{}{}{}{i}$ is well-defined in each $\G^i/d$ due to the back-and-forth conditions on 
$\kbisim{ii}$: Let $\TransitionD{[v_1]}{a}{[v_2]}{i}$, whence there are
$v'_1 \in [v_1]_{\kbisim{ii}}$ and $v'_2 \in [v_2]_{\kbisim{ii}}$ with $\Transition{v'_1}{a_i}{v'_2}$.
Then for any $v''_1 \in [v_1]_{\kbisim{ii}}$ and $v''_2 \in [v_2]_{\kbisim{ii}}$, we have
$v''_1 \kbisim{ii} v_1 \kbisim{ii} v''_1$, and $v''_2 \kbisim{ii} v_2 \kbisim{ii} v''_2$.

The following shows that this construction correctly re-factors graphs that indeed are powers.
\begin{lemma}
Let $\G_0,\dotsc,G_{d-1}$ be $(\Sigma, \Cols)$-graphs. For all $i \in [d]$ we have 
$\G_i \sim (\prod_{i=0}^{d-1} \G_i)^i/d$.
\end{lemma}
\begin{proof}
It suffices to show that $\{(v, [(v_0,\dotsc,v_{d-1})]_{\kbisim{ii}}) \mid v_i \sim v\}$ is a bisimulation 
that includes $(v^i_I, [(v^0_I, \dotsc, v^{d-1}_I)]_{\kbisim{ii}})$ where $v^j_I$ is the initial
node of $\G_j$. This is rather straightforward.
\end{proof}

We now show that the above construction yields factors such that their product
is bisimilar to the original graph again.
\begin{lemma}
Let $\Sigma,\Cols,d$ be given and $\G = (V,\Transition{}{}{},L,v_I)$ be an $\Func{d}(\Sigma,\Cols)$-graph with
associated $d$-bisimulation $(\kbisim{ij})_{0 \leq i, j \leq d-1}$. If $(\kbisim{ij})_{0 \leq i, j \leq d-1}$ is persistent and has the reset property, then $\G \sim \prod_{i=0}^{d-1} \G^i/d$.
\end{lemma}


\begin{proof}

Let $\G = (V,\Transition{}{}{},L,v_I)$ and let $\G^i/d = (V/d, \TransitionD{}{}{}{i}, L^i/d, v^i_I)$. 
We write $\G' = (V', \TransitionD{}{}{}{\Pi}, L', v'_I)$ for $\prod_{i=0}^{d-1} \G^i/d$.

It suffices to show that the relation 
\[R = \{ (v, ([v]_{\kbisim{00}}, \ldots, [v]_{\kbisim{d-1\;d-1}})) \mid v \in V \}
\] is a bisimulation and that
$(v_I,([v_I]_{\kbisim{00}}, \ldots, [v_I]_{\kbisim{d-1\;d-1}})) \in R$. 
For the sake of readability, we drop the subscripts $[\cdot]_{\kbisim{jj}}$ from
the equivalence classes in places where there is no possibility
of confusion.

\Prop: For all states, we have that $v$ and $([v], \dotsc, [v])$ have the same labelling, since 
$c_i \in L'([v], \dotsc, [v])$ iff $c \in L^i_d([v]_{\kbisim{ii}})$ iff $c_i \in L(v)$.

\Forth: Assume that $\Transition{v}{a_i}{v'}$ in $\G$. We have to show that there
is an $a_i$-transition  in $\G'$ from $([v], \dotsc, [v])$
to $([v'], \dotsc, [v'])$.

By persistence, we have $v \kbisim{jj} v'$ (and $v' \kbisim{jj} v$) for all $j \not = i$, 
whence $[v]_{\kbisim{jj}} = [v']_{\kbisim{jj}}$ for all $j \not = i$. Hence, it suffices to observe
that $\TransitionD{[v]_{\kbisim{ii}}}{a}{[v']_{\kbisim{ii}}}{i}$ in $\G^i/d$ to obtain the desired result.

\Back: Conversely, assume that 
there is an $a_i$-transition from $([v], \dotsc, [v])$ to 
\begin{displaymath}
([v]_{\kbisim{00}}, \dotsc, [v]_{\kbisim{i-1\;i-1}}, [v']_{\kbisim{ii}}, [v]_{\kbisim{i+1\;i+1}},\dotsc, [v]_{\kbisim{d-1\;d-1}})
\end{displaymath}
in $\G'$. 
We have to show that $\Transition{v}{a_i}{v''}$ in $\G$ for some $v''$ such that $[v]_{\kbisim{jj}} = [v'']_{\kbisim{jj}}$ for all $j \not = i$, and $[v']_{\kbisim{ii}} = [v'']_{\kbisim{ii}}$. 

From the transition above we obtain that $\Transition{[v]_{\kbisim{ii}}}{a}{[v']_{\kbisim{ii}}}$ in $\G^i/_d$. Hence,
there is $v'' \in [v']_{\kbisim{ii}}$ s.t.\ $\Transition{v}{(a, i)}{v''}$ in $\G$. By persistence, $[v]_{\kbisim{jj}} = [v'']_{\kbisim{jj}}$ for 
all $j \not = i$. Since also $[v'']_{\kbisim{ii}} = [v']_{\kbisim{ii}}$, we are done.

\Forth (again, this time for the $\reset{i}$-transitions):
Now assume that $\Transition{v}{\reset{i}}{v'}$ in $\G$.  We have to show that
there is a $\reset{i}$-transition from$([v], \dotsc, [v])$
to $([v'], \dotsc, [v'])$
in $\G'$. By persistence of $\G$, we have $v \kbisim{jj} v'$ and, hence
$[v]_{\kbisim{jj}} = [v']_{\kbisim{jj}}$ for all $i \not =j$. By the
reset property, $v' \kbisim{ii} v_I$. Since in $\G'$, reset transitions
simply change the $i$th component to $v^i_I = [v_I]_{\kbisim{ii}}$, we obtain that $[v_I]_{\kbisim{ii}} = [v']_{\kbisim{ii}}$.

\Back (for $\reset{i}$-transitions): Conversely, assume that there is a $\reset{i}$-transition in $\G'$
from $([v], \dotsc, [v])$ to 
\[([v]_{\kbisim{00}}, \dotsc,([v]_{\kbisim{i-1\;i-1}},  [v']_{\kbisim{ii}}, [v]_{\kbisim{i+1\;i+1}} \dotsc, [v]_{\kbisim{d-1\;d-1}}).\]
We have to show that $\Transition{v}{\reset{i}}{v''}$ in $\G$
for some $v''$ with $[v, j]_{\kbisim{jj}} = [v'']_{\kbisim{jj}}$ for all $j \not = i$,
and $[v']_{\kbisim{ii}} = [v'']_{\kbisim{ii}}$. 
Due to the way that reset transitions are defined, we have $[v']_{\kbisim{ii}} = v_I^i = [v_I]_{\kbisim{ii}}$.
Consider the unique node $v''$ in $\G$ with $\Transition{v}{\reset{i}}{v''}$.
By persistence, we obtain that $[v]_{\kbisim{jj}} = [v'']_{\kbisim{jj}}$ for
all $j \not = i$, and by the reset property, $[v'']_{\kbisim{ii}} = [v_I]_{\kbisim{ii}}$. 
\end{proof}

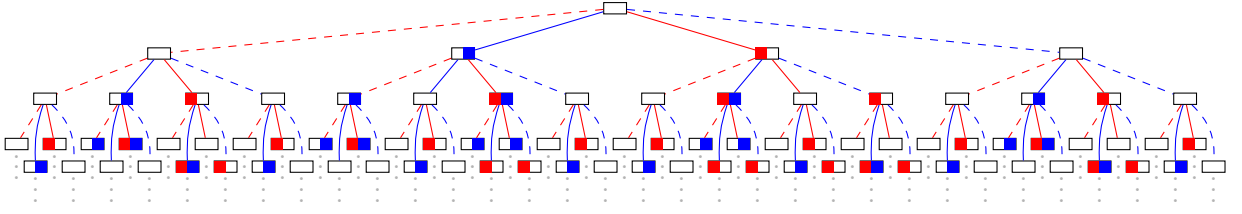
\begin{figure*}
\begin{center}
\begin{tikzpicture}[level distance=6mm,
                    level 1/.style={sibling distance=4cm},
                    level 2/.style={sibling distance=1cm},
                    level 3/.style={sibling distance=2.5mm},
                    level 4/.style={level distance=2mm}]
  \tikzstyle{a0}=[red,solid]
  \tikzstyle{a1}=[blue,solid]
  \tikzstyle{rst0}=[red,dashed]
  \tikzstyle{rst1}=[blue,dashed]
  \tikzstyle{f0}=[fill=red]
  \tikzstyle{f1}=[fill=blue]

  \newcommand{\vd}{\textcolor{gray!60}{\vdots}}
  
  \tikzstyle{second}=[edge from parent path={(\tikzparentnode\tikzparentanchor) edge [bend right=10] (\tikzchildnode\tikzchildanchor)}]
  \tikzstyle{fourth}=[edge from parent path={(\tikzparentnode\tikzparentanchor) edge [bend left=25] (\tikzchildnode\tikzchildanchor)}]
    
  \node[snode] (eps) {\fnone{}{}}
   child { node[snode] (r0-00) {\fnone{}{}}
           child { node[snode] (r0-00-r0-00) {\fnone{}{}}
                   child { node[snode] (r0-00-r0-00-r0-00) {\fnone{}{}}
                           child { node {$\vd$} edge from parent [draw=none] } 
                           edge from parent [rst0]
                         }
                   child { node[snode,yshift=-3mm] (r0-00-r0-00-a1-01) {\fone{}{}}
                           child { node {$\vd$} edge from parent [draw=none] } 
                           edge from parent [a1,second]
                         }
                   child { node[snode] (r0-00-r0-00-a0-10) {\fzero{}{}}
                           child { node {$\vd$} edge from parent [draw=none] } 
                           edge from parent [a0]
                         }
                   child { node[snode,yshift=-3mm] (r0-00-r0-00-r1-00) {\fnone{}{}}
                           child { node {$\vd$} edge from parent [draw=none] } 
                           edge from parent [rst1,fourth]
                         }
                   edge from parent [rst0]
                 }
           child { node[snode] (r0-00-a1-01) {\fone{}{}}
                   child { node[snode] (r0-00-a1-01-r0-01) {\fone{}{}}
                           child { node {$\vd$} edge from parent [draw=none] } 
                           edge from parent [rst0]
                         }
                   child { node[snode,yshift=-3mm] (r0-00-a1-01-a1-02) {\fnone{}{}}
                           child { node {$\vd$} edge from parent [draw=none] } 
                           edge from parent [a1,second]
                         }
                   child { node[snode] (r0-00-a1-01-a0-11) {\fboth{}{}}
                           child { node {$\vd$} edge from parent [draw=none] } 
                           edge from parent [a0]
                         }
                   child { node[snode,yshift=-3mm] (r0-00-a1-01-r1-00) {\fnone{}{}}
                           child { node {$\vd$} edge from parent [draw=none] } 
                           edge from parent [rst1,fourth]
                         }
                   edge from parent [a1]
                 }
           child { node[snode] (r0-00-a0-10) {\fzero{}{}}
                   child { node[snode] (r0-00-a0-10-r0-00) {\fnone{}{}}
                           child { node {$\vd$} edge from parent [draw=none] } 
                           edge from parent [rst0]
                         }
                   child { node[snode,yshift=-3mm] (r0-00-a0-10-a1-11) {\fboth{}{}}
                           child { node {$\vd$} edge from parent [draw=none] } 
                           edge from parent [a1,second]
                         }
                   child { node[snode] (r0-00-a0-10-a0-20) {\fnone{}{}}
                           child { node {$\vd$} edge from parent [draw=none] } 
                           edge from parent [a0]
                         }
                   child { node[snode,yshift=-3mm] (r0-00-a0-10-r1-10) {\fzero{}{}}
                           child { node {$\vd$} edge from parent [draw=none] } 
                           edge from parent [rst1,fourth]
                         }
                   edge from parent [a0]
                 }
           child { node[snode] (r0-00-r1-00) {\fnone{}{}}
                   child { node[snode] (r0-00-r1-00-r0-00) {\fnone{}{}}
                           child { node {$\vd$} edge from parent [draw=none] } 
                           edge from parent [rst0]
                         }
                   child { node[snode,yshift=-3mm] (r0-00-r1-00-a1-01) {\fone{}{}}
                           child { node {$\vd$} edge from parent [draw=none] } 
                           edge from parent [a1,second]
                         }
                   child { node[snode] (r0-00-r1-00-a0-10) {\fzero{}{}}
                           child { node {$\vd$} edge from parent [draw=none] } 
                           edge from parent [a0]
                         }
                   child { node[snode,yshift=-3mm] (r0-00-r1-00-r1-00) {\fnone{}{}}
                           child { node {$\vd$} edge from parent [draw=none] } 
                           edge from parent [rst1,fourth]
                         }
                   edge from parent [rst1]
                 }
           edge from parent [rst0]
         }
   child { node[snode] (a1-01) {\fone{}{}}
           child { node[snode] (a1-01-r0-01) {\fone{}{}}
                   child { node[snode] (a1-01-r0-01-r0-01) {\fone{}{}}
                           child { node {$\vd$} edge from parent [draw=none] } 
                           edge from parent [rst0]
                         }
                   child { node[snode,yshift=-3mm] (a1-01-r0-01-a1-02) {\fnone{}{}}
                           child { node {$\vd$} edge from parent [draw=none] } 
                           edge from parent [a1,second]
                         }
                   child { node[snode] (a1-01-r0-01-a0-11) {\fboth{}{}}
                           child { node {$\vd$} edge from parent [draw=none] } 
                           edge from parent [a0]
                         }
                   child { node[snode,yshift=-3mm] (a1-01-r0-01-r1-00) {\fnone{}{}}
                           child { node {$\vd$} edge from parent [draw=none] } 
                           edge from parent [rst1,fourth]
                         }
                   edge from parent [rst0]
                 }
           child { node[snode] (a1-01-a1-02) {\fnone{}{}}
                   child { node[snode] (a1-01-a1-02-r0-02) {\fnone{}{}}
                           child { node {$\vd$} edge from parent [draw=none] } 
                           edge from parent [rst0]
                         }
                   child { node[snode,yshift=-3mm] (a1-01-a1-02-a1-01) {\fone{}{}}
                           child { node {$\vd$} edge from parent [draw=none] } 
                           edge from parent [a1,second]
                         }
                   child { node[snode] (a1-01-a1-02-a0-12) {\fzero{}{}}
                           child { node {$\vd$} edge from parent [draw=none] } 
                           edge from parent [a0]
                         }
                   child { node[snode,yshift=-3mm] (a1-01-a1-02-r1-00) {\fnone{}{}}
                           child { node {$\vd$} edge from parent [draw=none] } 
                           edge from parent [rst1,fourth]
                         }
                   edge from parent [a1]
                 }
           child { node[snode] (a1-01-a0-11) {\fboth{}{}}
                   child { node[snode] (a1-01-a0-11-r0-01) {\fone{}{}}
                           child { node {$\vd$} edge from parent [draw=none] } 
                           edge from parent [rst0]
                         }
                   child { node[snode,yshift=-3mm] (a1-01-a0-11-a1-12) {\fzero{}{}}
                           child { node {$\vd$} edge from parent [draw=none] } 
                           edge from parent [a1,second]
                         }
                   child { node[snode] (a1-01-a0-11-a0-21) {\fone{}{}}
                           child { node {$\vd$} edge from parent [draw=none] } 
                           edge from parent [a0]
                         }
                   child { node[snode,yshift=-3mm] (a1-01-a0-11-r1-10) {\fzero{}{}}
                           child { node {$\vd$} edge from parent [draw=none] } 
                           edge from parent [rst1,fourth]
                         }
                   edge from parent [a0]
                 }
           child { node[snode] (a1-01-r1-00) {\fnone{}{}}
                   child { node[snode] (a1-01-r1-00-r0-00) {\fnone{}{}}
                           child { node {$\vd$} edge from parent [draw=none] } 
                           edge from parent [rst0]
                         }
                   child { node[snode,yshift=-3mm] (a1-01-r1-00-a1-01) {\fone{}{}}
                           child { node {$\vd$} edge from parent [draw=none] } 
                           edge from parent [a1,second]
                         }
                   child { node[snode] (a1-01-r1-00-a0-10) {\fzero{}{}}
                           child { node {$\vd$} edge from parent [draw=none] } 
                           edge from parent [a0]
                         }
                   child { node[snode,yshift=-3mm] (a1-01-r1-00-r1-00) {\fnone{}{}}
                           child { node {$\vd$} edge from parent [draw=none] } 
                           edge from parent [rst1,fourth]
                         }
                   edge from parent [rst1]
                 }
           edge from parent [a1]
         }
   child { node[snode] (a0-10) {\fzero{}{}}
           child { node[snode] (a0-10-r0-00) {\fnone{}{}}
                   child { node[snode] (a0-10-r0-00-r0-00) {\fnone{}{}}
                           child { node {$\vd$} edge from parent [draw=none] } 
                           edge from parent [rst0]
                         }
                   child { node[snode,yshift=-3mm] (a0-10-r0-00-a1-01) {\fone{}{}}
                           child { node {$\vd$} edge from parent [draw=none] } 
                           edge from parent [a1,second]
                         }
                   child { node[snode] (a0-10-r0-00-a0-10) {\fzero{}{}}
                           child { node {$\vd$} edge from parent [draw=none] } 
                           edge from parent [a0]
                         }
                   child { node[snode,yshift=-3mm] (a0-10-r0-00-r1-00) {\fnone{}{}}
                           child { node {$\vd$} edge from parent [draw=none] } 
                           edge from parent [rst1,fourth]
                         }
                   edge from parent [rst0]
                 }
           child { node[snode] (a0-10-a1-11) {\fboth{}{}}
                   child { node[snode] (a0-10-a1-11-r0-01) {\fone{}{}}
                           child { node {$\vd$} edge from parent [draw=none] } 
                           edge from parent [rst0]
                         }
                   child { node[snode,yshift=-3mm] (a0-10-a1-11-a1-12) {\fzero{}{}}
                           child { node {$\vd$} edge from parent [draw=none] } 
                           edge from parent [a1,second]
                         }
                   child { node[snode] (a0-10-a1-11-a0-21) {\fone{}{}}
                           child { node {$\vd$} edge from parent [draw=none] } 
                           edge from parent [a0]
                         }
                   child { node[snode,yshift=-3mm] (a0-10-a1-11-r1-10) {\fzero{}{}}
                           child { node {$\vd$} edge from parent [draw=none] } 
                           edge from parent [rst1,fourth]
                         }
                   edge from parent [a1]
                 }
           child { node[snode] (a0-10-a0-20) {\fnone{}{}}
                   child { node[snode] (a0-10-a0-20-r0-00) {\fnone{}{}}
                           child { node {$\vd$} edge from parent [draw=none] } 
                           edge from parent [rst0]
                         }
                   child { node[snode,yshift=-3mm] (a0-10-a0-20-a1-21) {\fone{}{}}
                           child { node {$\vd$} edge from parent [draw=none] } 
                           edge from parent [a1,second]
                         }
                   child { node[snode] (a0-10-a0-20-a0-10) {\fzero{}{}}
                           child { node {$\vd$} edge from parent [draw=none] } 
                           edge from parent [a0]
                         }
                   child { node[snode,yshift=-3mm] (a0-10-a0-20-r1-20) {\fzero{}{}}
                           child { node {$\vd$} edge from parent [draw=none] } 
                           edge from parent [rst1,fourth]
                         }
                   edge from parent [a0]
                 }
           child { node[snode] (a0-10-r1-10) {\fzero{}{}}
                   child { node[snode] (a0-10-r1-10-r0-00) {\fnone{}{}}
                           child { node {$\vd$} edge from parent [draw=none] } 
                           edge from parent [rst0]
                         }
                   child { node[snode,yshift=-3mm] (a0-10-r1-10-a1-11) {\fboth{}{}}
                           child { node {$\vd$} edge from parent [draw=none] } 
                           edge from parent [a1,second]
                         }
                   child { node[snode] (a0-10-r1-10-a0-20) {\fnone{}{}}
                           child { node {$\vd$} edge from parent [draw=none] } 
                           edge from parent [a0]
                         }
                   child { node[snode,yshift=-3mm] (a0-10-r1-10-r1-10) {\fzero{}{}}
                           child { node {$\vd$} edge from parent [draw=none] } 
                           edge from parent [rst1,fourth]
                         }
                   edge from parent [rst1]
                 }
           edge from parent [a0]
         }
   child { node[snode] (r1-00) {\fnone{}{}}
           child { node[snode] (r1-00-r0-00) {\fnone{}{}}
                   child { node[snode] (r1-00-r0-00-r0-00) {\fnone{}{}}
                           child { node {$\vd$} edge from parent [draw=none] } 
                           edge from parent [rst0]
                         }
                   child { node[snode,yshift=-3mm] (r1-00-r0-00-a1-01) {\fone{}{}}
                           child { node {$\vd$} edge from parent [draw=none] } 
                           edge from parent [a1,second]
                         }
                   child { node[snode] (r1-00-r0-00-a0-10) {\fzero{}{}}
                           child { node {$\vd$} edge from parent [draw=none] } 
                           edge from parent [a0]
                         }
                   child { node[snode,yshift=-3mm] (r1-00-r0-00-r1-00) {\fnone{}{}}
                           child { node {$\vd$} edge from parent [draw=none] } 
                           edge from parent [rst1,fourth]
                         }
                   edge from parent [rst0]
                 }
           child { node[snode] (r1-00-a1-01) {\fone{}{}}
                   child { node[snode] (r1-00-a1-01-r0-01) {\fone{}{}}
                           child { node {$\vd$} edge from parent [draw=none] } 
                           edge from parent [rst0]
                         }
                   child { node[snode,yshift=-3mm] (r1-00-a1-01-a1-02) {\fnone{}{}}
                           child { node {$\vd$} edge from parent [draw=none] } 
                           edge from parent [a1,second]
                         }
                   child { node[snode] (r1-00-a1-01-a0-11) {\fboth{}{}}
                           child { node {$\vd$} edge from parent [draw=none] } 
                           edge from parent [a0]
                         }
                   child { node[snode,yshift=-3mm] (r1-00-a1-01-r1-00) {\fnone{}{}}
                           child { node {$\vd$} edge from parent [draw=none] } 
                           edge from parent [rst1,fourth]
                         }
                   edge from parent [a1]
                 }
           child { node[snode] (r1-00-a0-10) {\fzero{}{}}
                   child { node[snode] (r1-00-a0-10-r0-00) {\fnone{}{}}
                           child { node {$\vd$} edge from parent [draw=none] } 
                           edge from parent [rst0]
                         }
                   child { node[snode,yshift=-3mm] (r1-00-a0-10-a1-11) {\fboth{}{}}
                           child { node {$\vd$} edge from parent [draw=none] } 
                           edge from parent [a1,second]
                         }
                   child { node[snode] (r1-00-a0-10-a0-20) {\fnone{}{}}
                           child { node {$\vd$} edge from parent [draw=none] } 
                           edge from parent [a0]
                         }
                   child { node[snode,yshift=-3mm] (r1-00-a0-10-r1-10) {\fzero{}{}}
                           child { node {$\vd$} edge from parent [draw=none] } 
                           edge from parent [rst1,fourth]
                         }
                   edge from parent [a0]
                 }
           child { node[snode] (r1-00-r1-00) {\fnone{}{}}
                   child { node[snode] (r1-00-r1-00-r0-00) {\fnone{}{}}
                           child { node {$\vd$} edge from parent [draw=none] } 
                           edge from parent [rst0]
                         }
                   child { node[snode,yshift=-3mm] (r1-00-r1-00-a1-01) {\fone{}{}}
                           child { node {$\vd$} edge from parent [draw=none] } 
                           edge from parent [a1,second]
                         }
                   child { node[snode] (r1-00-r1-00-a0-10) {\fzero{}{}}
                           child { node {$\vd$} edge from parent [draw=none] } 
                           edge from parent [a0]
                         }
                   child { node[snode,yshift=-3mm] (r1-00-r1-00-r1-00) {\fnone{}{}}
                           child { node {$\vd$} edge from parent [draw=none] } 
                           edge from parent [rst1,fourth]
                         }
                   edge from parent [rst1]
                 }
           edge from parent [rst1]
         }
  ; 


  
%

\end{tikzpicture}
\end{center}
\caption{Tree unfolding of the graph $\G^2$ from Fig.~\ref{fig:power}.}
\label{fig:dbisimulationtree}
\end{figure*}

\begin{example}
Reconsider the $(\{a\},\{f\})$-graph $\G$ from Ex.~\ref{exm:power} and its $2$-power $\G^2$ shown 
in Fig.~\ref{fig:power}. Its tree unfolding $(\G^2)^{\mathsf{unf}}$ is shown in Fig.~\ref{fig:dbisimulationtree}.
We omit arrow tips; all arrows point downwards in these trees as usual.

While the $2$-bisimulation $(\kbisim{ij})_{i,j \in [2]}$ exists on $\G^2$, we are ultimately interested in
characterisations of tree languages, with tree unfoldings of power graphs playing an important role. However,
it is impossible to draw $(\kbisim{ij})_{i,j \in [2]}$ on $(\G^2)^{\mathsf{unf}}$ without covering the entire
shown part of the tree in a blob of edges. Even on the 85 nodes alone that are shown of that tree, relation
$\kbisim{01}$ for instance contains $27^2 + 58^2 = 3645$ pairs already.

However, in this particular case, it is not difficult to describe these relations. The reason is of course
the simplicity of the underlying graph $\G$, in particular the fact that its bisimulation quotient partitions
the three nodes into two classes, and these are already distinguished by the atomic formula $f$, resp.\ by the
red colour on the left and the blue colour on the right. So,
\begin{itemize}
\item relation $\kbisim{00}$ contains all pairs of nodes $(u,v)$ whose left parts are coloured in the same
      way,
\item likewise, relation $\kbisim{11}$ contains all such pairs in which the right parts are coloured 
      in the same way,
\item relation $\kbisim{01}$ contains all pairs $(u,v)$ such that the left part of $u$ is red if and
      only if the right part of $v$ is blue, and
\item relation $\kbisim{10}$ is necessarily the inverse of $\kbisim{01}$.
\end{itemize}
Since all four relations can be seen -- in this particular case -- as relations on the four node types 
$\tikz \node[tnode] {\fnone{}{}};$, $\tikz \node[tnode] {\fzero{}{}};$,
$\tikz \node[tnode] {\fone{}{}};$ and $\tikz \node[tnode] {\fboth{}{}};$,
we can depict them more easily as a graph on these types, as follows.
\begin{center}
\begin{tikzpicture}[node distance=2cm]
  \node[tnode]             (ww) {\fnone{}{}};
  \node[tnode,right of=ww] (wb) {\fone{}{}};
  \node[tnode,below of=ww] (rw) {\fzero{}{}};
  \node[tnode,right of=rw] (rb) {\fboth{}{}};
  
  \path[semithick,->] 
    (ww) edge [out=120,in=170,looseness=10] node [above]            {\scriptsize $00,01,10,11,$} (ww) 
         edge [bend left=15]                node [pos=.2,right=-2pt]       {\scriptsize $01,11$}    (rw)
         edge [bend left=15]                node [above=-2pt]       {\scriptsize $00,10$}    (wb)
    (wb) edge [out=60,in=20,looseness=10]   node [above]            {\scriptsize $00,11$} (wb) 
         edge [bend left=15]                node [above=-2pt]       {\scriptsize $00,01$} (ww)
         edge [bend left=15]                node [sloped,below right=-2pt] {\scriptsize $01,10$} (rw)
         edge [bend left=15]                node [right=-2pt]       {\scriptsize $10,11$} (rb)
    (rw) edge [out=240,in=200,looseness=10] node [below]            {\scriptsize $00,11$} (rw) 
         edge [bend left=15]                node [left=-2pt]        {\scriptsize $10,11$}    (ww)
         edge [bend left=15]                node [sloped,above left=-2pt]  {\scriptsize $01,10$} (wb)
         edge [bend left=15]                node [below=-2pt]       {\scriptsize $00,01$} (rb)
    (rb) edge [out=300,in=340,looseness=10] node [below]            {\scriptsize $00,01,10,11$} (rb)
         edge [bend left=15]                node [below=-2pt]       {\scriptsize $00,10$} (rw)
         edge [bend left=15]                node [pos=.2,left=-2pt]        {\scriptsize $01,11$} (wb)
                      ;
\end{tikzpicture}
\end{center}
\end{example}


\subsection{Definability in $\polymucalc$}

We provide a second characterisation of power graphs, showing that the property of being a $d$-product,
and in fact a $d$-power, is in fact definable in the polyadic $\mu$-calculus itself. In fact, this is
possible in the \emph{dyadic} $\mu$-calculus $\mucalc{2}$ already.

Take a $\Func{d}(\Sigma,\Cols)$-graph $\G$ for some $\Sigma,\Cols$. Note that it has actions and atomic 
propositions of the form $a_i$, resp.\ $c_i$ for $i \in [d]$. Now, consider interpreting a \polymucalc-formula 
of arity $d'$ over $\G$. Its atomic formulas are of the form $c'_j$ for $j \in [d']$ and $c'$ of the form
$c_i$. Hence, in order to avoid confusion with the double indices, we write such a formula as $(c_i)_j$.
For actions, confusion should not arise in the first place, as the syntax prescribes formulas of the form
$\mudiam{a_i}{j}$ in such cases. However, when monofying such formulas according to Prop.~\ref{prop:langelozes},
we would also obtain actions with double indices which we would also write as $(a_i)_j$.
 
Let $\Sigma,\Cols, d \ge 1$ and $i,j \in [d]$ be given. Consider the $\mucalc{2}$-formula 
\begin{align*}
\varphi^{i,j}_{\mathsf{bis}} = \nu X_{i,j}.\ &\Big(\bigwedge_{c \in \Cols} (c_i)_0 \leftrightarrow (c_j)_1\Big) \ \wedge \ \\
&\Big(\bigwedge_{a \in \Sigma}  (\mubox{a_i}{0} \mudiam{a_j}{1} X_{i,j}) \wedge (\mubox{a_j}{1}\mudiam{a_i}{0} X_{i,j})\Big)
\end{align*}
which is readily seen to define $\kbisim{ij}$ as a greatest fixpoint
alongside the standard pattern for bisimulation (cf.\ \cite{Otto/99b}). Here, the $0$-th
component tracks the behaviour of a node w.r.t.\ to (base) transitions and colours
of the form $a_i$ and $c_i$, and the $1$-st component tracks behaviour for
$a_j$ and $c_j$.

We write $\allbox{i} \varphi$ for the formula
\begin{displaymath}
\nu X.\ \varphi \wedge \bigwedge_{a \in \Sigma} \bigwedge_{j=0}^{d-1} \mubox{a_j}{i} X
\end{displaymath}
which simply expresses that formula $\varphi$ holds for every possible (i.e., reachable) value of the $i$th
component. Again, cf.\ \cite{Otto/99b} for this use of universal quantification.

Definability of $d$-products and -powers of graphs hinges on the following three $\mucalc{2}$-formulas.
\begin{align*}
\varphi_{\mathsf{per}} &= \allbox{0}\allbox{1} \bigwedge_{i=0}^{d-1} \Big( \varphi^{j,j}_{\mathsf{bis}} \rightarrow \bigwedge_{\substack{j=0 \\ j \not =i}}^{d-1} \big(\mubox{a_i}{0} \varphi^{j,j}_{\mathsf{bis}}\big)  \wedge \big(\mubox{\reset{i}}{0}  \varphi^{j,j}_{\mathsf{bis}}\big)\Big) \\ 
\varphi_{\mathsf{rst}} &= \allbox{0} \bigwedge_{i=0}^{d-1} \mubox{\reset{i}}{0} \varphi^{i,i}_{\mathsf{bis}} \\
\varphi_{\mathsf{pow}} &= \bigwedge_{i=0}^{d-1}\bigwedge_{j=0}^{d-1} \varphi^{i,j}_{\mathsf{bis}}
\end{align*}

\begin{lemma}
\label{lem:powerchar2}
An $\Func{d}(\Sigma,\Cols)$-graph $\G$ is bisimilar to a $d$-product of $(\Sigma,\Cols)$ graphs 
iff $\G \models \varphi_{\mathsf{per}} \wedge \varphi_{\mathsf{rst}}$. 
It is bisimilar to the $d$-power of some $(\Sigma,\Cols)$-graph iff, additionally, 
$\G \models \varphi_{\mathsf{pow}}$.
\end{lemma}

\begin{proof}
It suffices to recognise the following. Formula $\varphi_{\mathsf{per}}$ enforces persistence: for
any tuple $(v_0, v_1)$ in $\kbisim{jj}$, it must be the case 
that taking any transition (of the form $a^i$ or $\reset{i}$, with $i \not = j$) in the $0$-th component yields a tuple 
$(v'_0, v_1)$ in $\kbisim{jj}$, i.e., taking a transition in the candidate $i$-th component 
anywhere leaves the other candidate components equivalent if they were so before.

Formula $\varphi_{\mathsf{rst}}$ enforces the reset property: for any tuple $(v, v_I)$, after taking a 
$\reset{k}$-transition in the $0$-th component, the resulting tuple satisfies $\kbisim{ii}$.

Formula $\varphi_{\mathsf{pow}}$ demands that $(v_I, v_I)$ in $\kbisim{ij}$ for all $i,j$ and the initial state 
$v_I$ of a $d$-product. This then means that all its $d$ components are mutually bisimilar, i.e.\ the 
$d$ factors of the product are all bisimilar, and the $d$-product is in fact a $d$-power. 
\end{proof}

An immediate consequence of Lemma~\ref{lem:powerchar2}, namely the fact that $d$-powers are definable
in a fragment of \polymucalc of fixed arity, independent of $d$, is the following, making use of
Prop.~\ref{prop:otto}.

\begin{corollary}
The property of a $\Func{d}(\Sigma,\Cols)$-graph being the $d$-power of some $(\Sigma,\Cols)$-graph is 
decidable in polynomial time.
\end{corollary}

%% file: separate.tex

We study the structure of tree languages that result from unfoldings of $d$-powers of graphs satisfying particular queries
which are known to be in \bisim{\NPTIME}, resp.\ \bisim{\PSPACE}. These are the aforementioned universality problems for
NFA over a 1-letter, resp.\ 2-letter alphabet, as used in the proof of Thm.~\ref{thm:bisseparate}.

\subsection{A Construction for \NPTIME}

Let $\Sigma = \{a\}$ and $\Cols = \{f\}$, i.e.\ here we are considering $(1,1)$-graphs, in particular the set 
$\OneNonUnivNFA := \{ \G = (V,\Transition{}{}{},L,v_I) \mid \exists n \ge 0\, \forall v \in V:$ if 
$\Transition{v_I}{a^n}{v}$ then $f \not\in L(v) \}$ where, as usual, $\Transition{}{a^n}{}$ is the $n$-fold product
of $\Transition{}{a}{}$. Note that this is indeed just a formalisation of the non-universality for NFA over the 
singleton alphabet $\Sigma$ where accepting states are coloured with $f$: is there a word (length) such that all 
paths under this word starting in the initial state end in non-accepting states?

\begin{proposition}[\cite{STOC73*1,Milner80,Park81:conais}]
\label{prop:1NonUnivNFAinNPbisim}
\enspace
\begin{enumerate}
\item[a)] $\OneNonUnivNFA{} \in \bisim{\NPTIME}$.
\item[b)] $\OneNonUnivNFA{}$ is \NPTIME-complete.
\end{enumerate}
\end{proposition} 
 
Inclusion in \NPTIME is usually argued for as follows. Using a pumping argument one obtains an exponential upper bound 
on the minimal $n$ such that all nodes reachable under $a^n$ are not coloured. A nondeterministic algorithm can then 
guess a binary representation of this number $n$ of polynomial size and construct the relation $\Transition{}{a^n}{}$
in polynomial time using iterated squaring. It can be stored as an adjacency matrix of quadratic size for example. Once
this matrix is computed, it is easy to derive all nodes reachable from the initial state under this relation and check
their colouring.

At last, \OneNonUnivNFA{} is also bisimulation-invariant: given two NFA $\G, \G'$ such that $\G \sim \G'$, suppose that
$\G \in \OneNonUnivNFA{}$. Bisimilarity is known to preserve trace equivalence, cf.\ \cite{Milner80,Park81:conais}, i.e.\ 
the languages of the two NFA must be equal. Thus, $\G' \in \OneNonUnivNFA{}$ as well which means that not only
$\OneNonUnivNFA{} \in \NPTIME$ but indeed $\OneNonUnivNFA{} \in \bisim{\NPTIME}$.

Now let $d \ge 1$. Remember that $\Func{d}(\{a\},\{f\})$-trees potentially represent unfoldings of $d$-powers of
$(\{a\},\{f\})$-graphs, i.e.\ NFA over a one-letter alphabet. Our aim is to characterise those structurally as
tree languages. 

\begin{definition}
\label{def:OneNonUnivNFAd}
For any $d \ge 1$ let $\OneNonUnivNFA{d}$ be the set of all $\Func{d}(\{a\},\{f\})$-graphs $\G = (V,\Transition{}{}{},L,v_I)$
for which there is some $n \ge 0$ such that for every path 
\begin{displaymath}
\Transition{v_I}{x_1}{}\Transition{v_1}{x_2}{}\Transition{v_2}{x_3}{\ldots}\Transition{}{x_m}{v_m}
\end{displaymath}
and every $i \in [d]$ the following holds. If the subsequence of $x_1,\ldots,x_m$ consisting of all actions 
$a_i$ that occur after the last action $\reset{i}$ (if it occurs) has length $n$, then $f_i \not\in L(v_m)$.
\end{definition}

\begin{lemma}
\label{lem:OneNonUnivNFAd}
For all $d \ge 1$ and all $(\{a\},\{f\})$-graphs $\G$ we have $\G^d \in \OneNonUnivNFA{d}$ iff 
$\G \in \OneNonUnivNFA{}$.
\end{lemma}

\begin{proof}
Assume that $\G = (V,\Transition{}{}{},L,v_I) \in \OneNonUnivNFA{}$. Then there is $n \geq 0$ such that, for all 
$v_n \in V$ with 
\begin{displaymath}
\Transition{v_I}{a}{}\Transition{v_1}{a}{}\Transition{v_2}{a}{\ldots}\Transition{}{a}{v_n}
\end{displaymath}
where $a$ is the unique edge label of $\G$, we have $f \in L(v_n)$. 

Let $d \geq 1$. Then $\G^d = (V^d, \TransitionD{}{}{}{d}, L^d, (v_I,\dotsc,v_I))$.
Let
\[\pi = \TransitionD{\bar{v_0}}{x_1}{}{d}\TransitionD{\bar{v_1}}{x_2}{}{d}\TransitionD{\bar{v_2}}{x_3}{}{d}{\ldots}\TransitionD{}{x_m}{\bar{v_m}}{d}\]
be any sequence of transitions in $\G^d$ such that $\bar{v_0} = (v_I, \dotsc, v_I)$. 
Let $i \in [d]$ and let $k \geq 0$ be the smallest number such
that $x_j \not  = \reset{i}$ for all $j \geq k$. 
Assume that $|\{k' > k \mid x_k = a_i\}| = n$, i.e.\ after the last
$\reset{i}$-transition (if it exists at all) in $\pi$, there are exactly $n$ transitions in the 
$i$-th component. Let $\pi_k$ denote
the subsequence of $\pi$ starting at node $v_{k}$, and 
let $v'_{k},\dotsc,v'_m$ be such that $v'_{k'} = (\bar{v_{k'}})_i$,
i.e.\ this is the sequence of $i$-th components of the tuples in $\pi_k$,
starting from index $k$. 

Observe that $v'_k = v_I$ necessarily, and that $v'_{k'} = v'_{k' + 1}$ if 
$x_{k'} \not = a_i$. Since there are exactly $n$ transitions
of the form $a_i$ in $\pi_k$, let $v''_0,\dotsc, v''_n$ be 
the unique subsequence of $v'_{k},\dotsc,v'_m$ such that $v''_0 = v'_k = v_I$
and $v''_{j} = v'_{k'+1}$ is such that $\TransitionD{v'_{k'}}{a_i}{v'_{k'+1}}{d}$,
i.e.\ this is the subsequence obtained by taking all nodes directly after
an $a_i$-transition. Note that $v''_n = (\bar{v_m})_i$ since this
component does not change after the last $a_i$-transition.

It is immediate that $\Transition{v''_0}{a}{}\Transition{v''_1}{a}{}\ldots\Transition{}{a}{}\Transition{v''_{n-1}}{a}{v_m}$
in $\G$, i.e.\ $f \in L(v_m)$ and, hence, $f_i \in L^d(\bar{v_m})$. Since
$\pi$ was arbitrary, this shows that $\G^d \in \OneNonUnivNFA{d}$.

Conversely, let $\G^d \in \OneNonUnivNFA{d}$, and let $n$ be as in the definition of $\OneNonUnivNFA{d}$.
It suffices to show that, for all paths $\Transition{v_I}{a}{}\Transition{v_1}{a}{}\dots\Transition{v_{n-1}}{a}{v_n}$
in $\G$, we have that $f \in L(v_n)$. But such a path gives rise to a sequence
\[\TransitionD{(v_I,\dotsc,v_I)}{a_0}{}{d}\TransitionD{(v_1, v_I, \dotsc,v_I)}{a_0}{}{d}\ldots\TransitionD{}{a_0}{(v_n, v_I, \dotsc,v_I)}{d}\] 
of transitions in $\G^d$ which clearly satisfies the conditions
lined out in the definition of $\OneNonUnivNFA{d}$. In particular,
it contains no $\reset{i}$-transitions at all. Hence, $f_i \in L^d((v_n, v_I, \dotsc,v_I))$
whence $f \in L(v_n)$. Consequently, $\G \in \OneNonUnivNFA{}$.
\end{proof}

\begin{theorem}
\label{thm:nonregforNP}
We have $\PTIME{\ne}\NPTIME$ iff for all $d \ge 1$, $\OneNonUnivNFA{d}$ is non-regular 
relative to $\textsc{Power}_{d}$.
\end{theorem}

\begin{proof}
For the ``if''-part suppose that $\PTIME{=}\NPTIME$. According to Thm.~\ref{thm:bisseparate}, we then also have 
$\bisim{\PTIME}{=}\bisim{\NPTIME}$. According to Prop.~\ref{prop:1NonUnivNFAinNPbisim}
we would have $\OneNonUnivNFA{} \in \bisim{\PTIME}$ in particular. According to Prop.~\ref{prop:otto}, there is some
$d \ge 1$ and a $d$-rooted $\mucalc{d+1}$-formula $\varphi_{\OneNonUnivNFA{}}$ s.t.\ 
$L(\varphi_{\OneNonUnivNFA{}}) = \OneNonUnivNFA{}$. Now consider its monofication  
$\mono{\varphi_{\OneNonUnivNFA{}}} \in \mucalc{}$ over $\Func{d}(\{a\},\{f\})$. By Prop.~\ref{prop:langelozes}, 
we have $\G^d \in L(\mono{\varphi_{\OneNonUnivNFA{}}})$ iff $\G \in \OneNonUnivNFA{}$ for any $(\{a\},\{f\})$-graph $\G$.
Hence, $\mono{\varphi_{\OneNonUnivNFA{}}} \in \mucalc{}$ defines the set of all trees that are bisimilar to 
$d$-powers of graphs in $\OneNonUnivNFA{}$ which is the tree language $\OneNonUnivNFA{d}$ according to 
Lemma~\ref{lem:OneNonUnivNFAd}. Thus, $\OneNonUnivNFA{d}$ is regular relative to $\textsc{Power}_d$.    

For the ``only if''-part suppose that $\OneNonUnivNFA{d}$, for some $d \ge 1$, is
a regular language of trees, relative to $\textsc{Power}_{d}$. I.e.\ there is some formula $\varphi \in \mucalc{}$
over $\Func{d}(\{a\},\{f\})$ which correctly identifies trees that belong to $\OneNonUnivNFA{d}$, provided that they
are bisimilar to the $d$-power of some $(\{a\},\{f\})$-graph $\G$. According to Lemma~\ref{lem:OneNonUnivNFAd}, this 
is the case only if $\G \in \OneNonUnivNFA{}$. According to Lemma~\ref{lem:polyfication}, the $d$-rooted formula
$\poly{\varphi} \in \mucalc{d+1}$ then defines the class $\OneNonUnivNFA{}$ of $(\{a\},\{f\})$-graphs. According
to Otto's Theorem (Prop.~\ref{prop:otto}), it is polynomial-time recognisable, i.e.\ it belongs to $\bisim{\PTIME}$
and therefore also to $\PTIME$. According to Prop.~\ref{prop:1NonUnivNFAinNPbisim} it is \NPTIME-hard, i.e.\ we
immediately get $\PTIME{=}\NPTIME$.
\end{proof}

\subsection{A Construction for \PSPACE}
\label{sec:sepPSPACE}

An argument can be formed along the same lines, characterising the separation of \PTIME from \PSPACE through the
existence of tree languages that are non-regular relative to the language of tree representations of power graphs, 
up to bisimilarity. All that is needed is a set of graphs that is bisimulation-invariant and \PSPACE-complete. A
candicate was been mentioned before: the (non-)universality problem for NFA over a two-letter alphabet. Hence, 
here we are dealing with $(\{a,b\},\{f\})$-graphs and their powers, etc.

\begin{proposition}[\cite{STOC73*1,Milner80,Park81:conais}]
\label{prop:2NonUnivNFAinPSPACEbisim}
\enspace
\begin{enumerate}
\item[a)] $\TwoNonUnivNFA{} \in \bisim{\PSPACE}$.
\item[b)] $\TwoNonUnivNFA{}$ is \PSPACE-complete.
\end{enumerate}
\end{proposition} 

What remains to be done is to characterise -- along the lines of Def.~\ref{def:OneNonUnivNFAd} and 
Lemma~\ref{lem:OneNonUnivNFAd} -- the set of $\Func{d}(\{a,b\},\{f\})$-trees that are bisimilar to the $d$-power 
of some graph in $\TwoNonUnivNFA{}$.

\begin{definition}
\label{def:TwoNonUnivNFAd}
For any $d \ge 1$ let $\TwoNonUnivNFA{d}$ be the set of all $\Func{d}(\{a\},\{f\})$-graphs $\G = (V,\Transition{}{}{},L,v_I)$
for which there is some $w \in \{a,b\}^*$ such that for every path 
\begin{displaymath}
\Transition{v_I}{x_1}{}\Transition{v_1}{x_2}{}\Transition{v_2}{x_3}{\ldots}\Transition{}{x_m}{v_m}
\end{displaymath}
and every $i \in [d]$ the following holds. If the subsequence of $x_1,\ldots,x_m$ consisting of all actions 
$a_i,b_i$ that occur after the last action $\reset{i}$ (if it occurs) equals $w$ when removing index $i$ on each letter, 
then $f \not\in L(v_m)$.
\end{definition}

\begin{lemma}
\label{lem:TwoNonUnivNFAd}
For all $d \ge 1$ and all $(\{a,b\},\{f\})$-graphs $\G$ we have $\G^d \in \TwoNonUnivNFA{d}$ iff 
$\G \in \TwoNonUnivNFA{}$.
\end{lemma}

With this characterisation, Prop.~\ref{prop:2NonUnivNFAinPSPACEbisim} and Thm.~\ref{thm:bisseparate}, we can 
repeat the argument in the proof of Thm.~\ref{thm:nonregforNP} using \TwoNonUnivNFA{} instead of \OneNonUnivNFA{} to
obtain the following.
 
\begin{theorem}
\label{thm:nonregforPSPACE}
We have $\PTIME{\ne}\PSPACE$ iff for all $d \ge 1$, $\TwoNonUnivNFA{d}$ is non-regular 
relative to $\textsc{Power}_{d}$.
\end{theorem}

The argument used in Thms.~\ref{thm:nonregforNP} and \ref{thm:nonregforPSPACE} is of course not restricted to the use 
of the 1-letter or 2-letter non-universality problem for NFA. Any \NPTIME-complete, resp.\ \PSPACE-complete problem that 
is bisimulation-invariant will also give rise to an infinite family of tree languages such that a proof of non-regularity 
of all of them, relative to the corresponding class $\textsc{Power}_{d}$ will separate \PTIME from \NPTIME, resp.\ \PSPACE. 
We therefore reformulate these theorems without the explicit reference to one such problem. 

\begin{corollary}
There are $\Sigma,\Cols$ and families $(L_d)_{d \ge 1}$ of languages of $\Func{d}(\Sigma,\Cols)$-trees 
such that $\PTIME{\ne}\NPTIME$, resp.\ $\PTIME{\ne}\PSPACE$ iff for all $d \ge 1$, $L_d$ is non-regular relative to 
$\textsc{Power}_d$.
\end{corollary}

%% file: concl.tex

\subsection{Summary}

We took two constructions from the literature on the polyadic $\mu$-calculus -- namely Otto's Theorem stating that 
it captures bisimula\-tion-invariant \PTIME, and a reduction of its model checking problem to that of the ordinary
modal $\mu$-calculus -- and studied the use of their combination in order to characterise the coincidence or separation 
of \bisim{\PTIME} from \bisim{\NPTIME}, resp.\ \bisim{\PSPACE}, which would equally entail the coincidence or separation
of the corresponding general complexity classes, and it does not suffer from the order problem. 

In the end, we obtained a characterisation of separation in terms
of non-regularity of all the members of some tree languages, relative to the language of trees that are bisimilar to
graphs that are obtained as powers in that product construction.

While this provides another line of attack regarding the \PTIME{=}\NPTIME problem (and the \PTIME{=}\PSPACE problem)
via the theory of formal tree languages and expressiveness therein, it is certainly fair to ask whether this can 
feasibly be used for such a proof of separation or coincidence between these classes. It is important to note that 
standard proofs of non-regularity do not suffice; instead candidate trees and their pumped versions need to be
constructed to remain with the class of trees bisimilar to power graphs. To this end, we provided an algebraic
characterisation of such graphs in terms of an extended notion of bisimulation, and a logical characterisation in
terms of definability in $\mucalc{2}$. This feels tantalisingly close to $\mucalc{1} = \mucalc{}$, but definability
in $\mucalc{}$ is of course equivalent to regularity in which case the notion of relative non-regularity would just
collapse to ordinary non-regularity. This would probably make such a combinatorial approach via pumping more 
feasible which is not to be expected. Regularity of the tree languages $\textsc{Power}_d$ would also be very
counterintuitive.

At this point it is perhaps worth mentioning the Abiteboul-Vianu Theorem stating that \PTIME{=}\PSPACE iff 
FO+LFP{=}FO+PFP \cite{Abiteboul87a}. Thus, it also provides a characterisation of the coincidence or separation
between the classes \PTIME and \PSPACE, there in terms of equi-expressiveness of two fixpoint logics. This
characterisation also circumvents the order problem in that it relates the \PTIME{=}\PSPACE question to the
question of expressiveness of FO+LFP and FO+PFP over general structures, not just ordered ones.

\subsection{Further Work}

From this formal-tree-language-based characterisation of the \PTIME{=}\NPTIME and \PTIME{=}\PSPACE problem, we
immediately obtain two obvious tasks for further research: (\textsc{i}) a better understanding of the structure 
of trees that are bisimilar to power graphs will be useful, should one even stand the chance to successfully use 
Lemma~\ref{lem:pumping} (or something similar) for proofs of non-regularity relative to the languages 
$\textsc{Power}_d$. (\textsc{ii}) Very much related to that, it will be useful to better understand the 
limits of expressivity of the dyadic fragment $\mucalc{2}$ of the polyadic $\mu$-calculus, as this may not only 
also shed more light onto the possibility to manipulate trees (by pumping etc.) within the class $\textsc{Power}_d$
but, conversely, it may provide insight on what is \emph{not} possible there.

At last, Fig.~\ref{fig:bisimcomplexity} points at another rather obvious line of future work: it is interesting
that descriptive complexity theory was started with a logical characterisation of the class \NPTIME, and then
others followed. On the bisimulation-invariant side, though, it is exactly the class \bisim{\NPTIME} for which
a logical characterisation is unknown so far, even though the bisimulation-invariant fragments of other major
time- and space-complexity classes do have such characterisations by now. Again, a capturing result for 
\bisim{\NPTIME} would open another possibility for attacking the \PTIME{=}\NPTIME question that is not impeded
by concerns about the existence or non-existence of orders.